\documentclass[final,twoside,11pt]{entics}

\usepackage{enticsmacro}
\usepackage{graphicx}
\usepackage[all]{xy}

\usepackage{subcaption} %

\usepackage{amsmath} 
\usepackage{array} 
\usepackage{textgreek} 
\usepackage{xspace} 
\usepackage{stmaryrd} 

\usepackage[crefname,capitalize]{lib/mymathpartir} 
\usepackage[bracketed=,macros]{lib/mymathtools}
\usepackage[bracketed=,macros]{lib/mycompounds}

\usepackage[capitalize,noabbrev,nameinlink]{cleveref} 
\usepackage{glossaries-extra} 

\usepackage{fontawesome} 

\input{lib/macros}

\DeclareSymbolFont{sfoperators}{T1}{\sfdefault}{\mddefault}{\itdefault}

\makeatletter
\renewcommand{\operator@font}{\mathgroup\symsfoperators}
\makeatother

\renewcommand{\termfont}{\mathsf}

\crefname{thm}{Theorem}{Theorems}
\Crefname{thm}{Theorem}{Theorems}
\crefname{prop}{Proposition}{Propositions}
\Crefname{prop}{Proposition}{Propositions}
\crefname{def}{Definition}{Definitions}
\Crefname{def}{Definition}{Definitions}
\crefname{cor}{Corollary}{Corollaries}
\Crefname{cor}{Corollary}{Corollaries}
\crefname{rule}{Rule}{Rules}
\Crefname{rule}{Rule}{Rules}
\crefname{lem}{Lemma}{Lemmas}
\Crefname{lem}{Lemma}{Lemmas}

\def\conf{MFPS 2026} 	
\volume{6}			
\def\volu{6}			
\def\papno{18}			
\def\mydoi{19685}
\def\lastname{Valliappan} 
\def\runtitle{Cover Semantics for Intuitionistic Modalities} 
\def\copynames{N. Valliappan}    
\begin{document}
\begin{frontmatter}
  \title{Cover Semantics for Intuitionistic Modalities} 
  \author{Nachiappan Valliappan\thanksref{myemail}}	
   \address{School of Informatics\\ University of Edinburgh\\				
    Edinburgh, Scotland}
   \thanks[myemail]{Email: \href{mailto:nachivpn@gmail.com} {\texttt{\normalshape
        nachivpn@gmail.com}}}
 \begin{abstract}
   Intuitionistic modal logic (IML) has inspired several developments
   in programming languages including modal type systems for staging,
   computational effects and language-based security. IMLs are
   typically studied using Kripke-style relational semantics, which
   simplifies proofs of meta-theoretic properties, such as
   completeness and consistency, by making it easy to construct
   models. Kripke-style relational semantics, however, relies upon
   classical reasoning principles, which makes it unappealing from a
   computational perspective and unsuitable for formalization in a
   constructive type theory. Goldblatt provides an alternative
   semantics for IMLs by extending Beth-Kripke-Joyal-style ``cover''
   semantics for intuitionistic propositional logic with relations to
   support modalities. Goldblatt's ``relational cover'' semantics
   overcomes classical reasoning but introduces a new limitation: it
   relies upon a ``modal localization'' condition that restricts the
   class of models and complicates model construction. Goldblatt
   bypasses this restriction by using intricate order-theoretic
   completion arguments to prove completeness. In this article, we
   present a conservative extension of relational cover semantics that
   alleviates this restriction and is amenable to simpler and standard
   model construction techniques. We formalize our semantics in Agda
   and prove completeness constructively in the style of Normalization
   by Evaluation for a variety of IMLs featuring independent box and
   diamond modalities.
\end{abstract}



\end{frontmatter}

\section{Introduction}\label{sec:intro}

Intuitionistic modal logic (IML) is the study of formal logics that
extend intuitionistic propositional logic with modalities such as the
box ($\BoxPr$) and diamond ($\DiaPr$) connectives.
Early work on IML can be found beginning with Fitch~\cite{Fitch48} in
the late 1940s, followed by pioneering contributions from
Fischer-Servi~\cite{Servi77,Servi81}, Božić and Došen~\cite{BozicD84},
Sotirov~\cite{Sotirov80}, and many others
since~\cite{PlotkinS86,Wijesekera90,Simpson94a}.
These studies have found various applications in computer science,
notably inspiring the design of modal type systems in programming
languages for distributed computing~\cite{MurphyCHP04},
meta\hyp{}programming~\cite{DaviesP01,NanevskiPP08}, guarded
recursion~\cite{BizjakGCMB16,BirkedalCMMPS20} and language\hyp{}based
security~\cite{GargP06,BorghuisF00,AbadiBHR99}.
Several recent
developments~\cite{TangWDHLL25,GouniPA25,LorenzenWDEL24,HuP24,Ahman23}
in modal type systems can be directly traced back to earlier
work~\cite{Borghuis94,PfenningW95,PfenningD01} on the proof theory and
natural deduction calculi for IMLs.

In contrast to the enthusiastic adoption of the proof theory for IMLs,
the model\hyp{}theoretic strengths of IMLs remain largely
under\hyp{}utilized in the study of programming languages.
This is an opportunity missed: modal logics enjoy a rich semantic
foundation with slick model construction techniques that could
simplify the way we currently reason about modal type systems.
The trouble, however, lies in the fact that most developments in the
semantics of IMLs rely upon classical reasoning principles, such as
proof by contradiction or the axiom of choice, which inhibits their
adoption in the study of programming languages.
The objective of this article is to develop a new semantics for IMLs
that does not require classical reasoning.

\textbf{Kripke\hyp{}style relational semantics}. The standard
semantics used to model IMLs extends Kripke's semantics for
\IPL~\cite{Kripke65} using an \emph{accessibility}
relation~\cite{Simpson94a}.
The truth of a formula is given using a triple~$F =
\triple{W}{\Ri{}{}}{\Rm}$ known as a \emph{frame}, which consists of a
set~$W$ of worlds, a partial order relation~$\Ri{}{}$ on worlds and an
accessibility relation~${\Rm} \subseteq W \times W$ subject to certain
compatibility conditions.
Given a model~$\Mod[M] = \tuple{F,V}$, consisting of a frame~$F$ and a
valuation~$V$ of propositional atoms, we say that a formula~$A$ is
true for a world~$w$ whenever the \emph{satisfaction}
relation~$\Mod[M], w \Vdash A$ holds.
The satisfaction relation is defined for an IML by extending the usual
definition for \IPL originally given by Kripke~\cite{Kripke65}.
In particular, satisfaction is defined for the \emph{positive},
i.e. falsity ($\BotPr$) and disjunction ($\OrPr$), connectives as:
\begin{equation*}
  \begin{array}{l@{\;\Vdash\;} @{\;} l @{\;\text{iff}\;}c@{\;} l}
    \Mod[M],w & \BotPr & & \text{false} \\
    \Mod[M],w & A \OrPr B & & \Mod[M],w \Vdash A\ \text{or}\ \Mod[M],w \Vdash B
  \end{array}
\end{equation*}

The satisfaction of modal formulas, typically $\BoxPr{A}$ and
$\DiaPr{A}$, is defined using the accessibility relation~$\Rm$ and can
vary significantly depending on the logic and applications under
consideration.
A comprehensive formal analysis of these variations can be found in a
recent survey of IMLs by De Groot et al.~\cite{GrootSC25}, who propose
a sweeping generalization of several common variants for boxes and
diamonds as:
\begin{equation*}
  \begin{array}{l@{\;\Vdash\;} @{\;} l @{\;\text{iff}\;}c@{\;} l}
    \Mod[M], w & \BoxPr{A} & & \forall w'.\, \Ri{w}{w'}\ \text{implies}\ \forall v.\, w' \Rm v\  \text{implies}\ \Mod[M],v \Vdash A \\
    \Mod[M], w & \DiaPr{A} & & \forall w'.\, \Ri{w}{w'}\ \text{implies}\ \exists v.\, w' \Rm v\  \text{and}\ \Mod[M],v \Vdash A
  \end{array}
\end{equation*}

Kripke\hyp{}style relational semantics has been used to model a wide
variety of IMLs for their remarkable ability to simplify proofs of
complex meta\hyp{}theoretic properties such as completeness and
consistency by making it easy to construct a model.
To construct a model, we need to identify four
parameters~$\tuple{W,\Ri{}{},\Rm,V}$ and show that they satisfy the
necessary compatibility conditions---a process that requires far less
ingenuity, for example, in comparison with algebraic models based on
Heyting algebras.
The price to pay, however, is that the proof of completeness, in
particular the canonical model construction, typically relies upon the
existence of \emph{prime filters}~that presumes availability of the
axiom of choice \cite[Remark 2.2]{Kavvos24a}.

\textbf{Relational cover semantics}.
Goldblatt~\cite{Goldblatt11a,Goldblatt11b} provides an alternative
semantics for IMLs by extending a re\hyp{}development of the
so\hyp{}called cover, or ``Kripke\hyp{}Joyal'', semantics for \IPL
using an accessibility relation for each modality in an IML.
Goldblatt's re\hyp{}development presents ideas that ``originated in
topos theory, in the logic of categories of sheaves'', typically
attributed to Joyal~\cite[Section 1]{Kock76} and Beth~\cite{Beth56},
``in a more general context that abstracts away from topological
spaces''.
A detailed account of this re\hyp{}development can be found in
\cite[Section 3]{Goldblatt11b}.
As opposed to a frame, the truth of a formula is given using a
\emph{cover system}~$C = \tuple{W,\Ri{}{},\Covp{}{}}$, consisting of a
partial order~$\tuple{W,\Ri{}{}}$ and a \emph{covering}
relation~$\Covp{}{}\ \subseteq W \times \Pow{(W)}$, accompanied by
accessibility relations such as $\RmBox$ and $\RmDia$ subject to
certain compatibility conditions.

The definition of the satisfaction relation departs notably from
Kripke\hyp{}style semantics for \emph{both} the positive connectives
and the modalities.
The clauses below define satisfaction for the positive connectives in
relational cover semantics by removing the ``immediacy'' that
Kripke\hyp{}style semantics necessitates.
\begin{equation*}
  \begin{array}{l@{\;\Vdash\;} @{\;} l @{\;\text{iff}\;}c@{\;} l}
    \Mod[M],w & \BotPr & & \Covp{w}{\emptyset} \\
    \Mod[M],w & A \OrPr B & & \exists \alpha.\, \Covp{w}{\alpha}\ \text{and}\ \forall v \in \alpha.\, \Mod[M],v \Vdash A\ \text{or}\ \Mod[M],v \Vdash B
    \end{array}
\end{equation*}
A formula~$A \OrPr B$ is true for world~$w$ iff either $A$ or $B$ is
true, not necessarily for $w$ itself, but for all worlds~$v$ in some
subset~$\alpha \subseteq W$ that \emph{covers} $w$.
Similarly, $\BotPr$ is true for a world~$w$ iff the empty set
$\emptyset$ covers~$w$.

The clauses below define satisfaction for the $\BoxPr$ and
$\DiaPr$ modalities in a somewhat unusual manner, in contrast to
Kripke\hyp{}style semantics, by treating both modalities like diamonds
in classical modal logic.
\begin{equation*}
  \begin{array}{l@{\;\Vdash\;} @{\;} l @{\;\text{iff}\;}c@{\;} l}
    \Mod[M], w & \BoxPr{A} & & \exists v.\, w \RmBox v\  \text{and}\ \Mod[M],v \Vdash A \\
    \Mod[M], w & \DiaPr{A} & & \exists v.\, w \RmDia v\  \text{and}\ \Mod[M],v \Vdash A
  \end{array}
\end{equation*}
Goldblatt treats all modalities alike under the slogan that ``there is
more to intuitionistic modal logic than the generalisation of
properties of boxes and diamonds from Boolean modal
logic''~\cite{Goldblatt11a}.
The logical properties of each individual modality is modeled by
imposing additional conditions on its respective accessibility
relation.
For example, the necessitation rule for the box modality (if $A$ is a
valid formula, then so is $\BoxPr{A}$) is modeled by requiring
$\RmBox$ to be a serial relation~\cite[Section 7]{Goldblatt11a}.
The result is a uniform semantics that models a variety of IMLs
including Bellin et al.'s~Constructive K (\CK)~\cite{BellinPR01},
Bierman and De Paiva's Constructive S4 (\CSFour)~\cite{BiermanP00} and
Fairtlough and Mendler's Propositional Lax Logic
(\LL)~\cite{FairtloughM97}.
These results readily extend further to weaker IMLs including
sublogics of \CK, namely Božić and Došen's~\CKBox
and~\CKDia~\cite{BozicD84}, and sublogics of \LL, namely the
logics~\SL,~\RL and~\JL~\cite{Valliappan26}.

A notable character of relational cover semantics is that it does not
demand classical reasoning since Goldblatt's completeness proofs do
not use prime filters.
Relational cover semantics, however, introduces a new problem: it no
longer supports standard model construction techniques used to
construct canonical models.
In an attempt to prove completeness by constructing a
Henkin\hyp{}style canonical model, Goldblatt encounters a ``stumbling
block''~\cite[Section 8]{Goldblatt11a} due to a condition imposed on
relational cover models known as \emph{modal localization}, which has
to do with an interaction between the relations~$\Rm$ and~$\Covp{}{}$.
Goldblatt bypasses this roadblock by resorting to the use of
\emph{MacNeille completion} to construct a different kind of model
that satisfies modal localization.
The details of Goldblatt's construction are rather intricate and more
involved than well\hyp{}known techniques used to prove completeness
for the cover semantics of \IPL---reasons which have likely inhibited
the larger adoption of relational cover semantics in IML literature.

\textbf{Modal cover semantics}.
In this article, we present a conservative extension of relational
cover semantics, which we shall call \emph{modal cover semantics}, by
replacing the accessibility relation~${\Rm} \subseteq W \times W$ that
accompanies a cover system in a relational cover model with a
\emph{modal covering}
relation~${\Covm{}{}} \subseteq W \times \Pow{(W)}$.
For a world~$w$, the modal covering relation relaxes the concept of a
possible ``future'' world~$v \in W$, given by the
relationship~$w \Rm v$, to a collection or ``neighborhood'' of
possible future worlds~$\alpha \subseteq W$, given by the
relationship~$\Covm{w}{\alpha}$.
The resulting semantics retains the convenience of model construction
in Kripke\hyp{}style relational semantics while continuing to avoid
classical reasoning as in relational cover semantics.
We show that modal cover semantics can be used to model five
important IMLs, which are namely:
\begin{enumerate}
\item a minimal monotone modal logic~\CM featuring a modality~$\MonPr$
  that generalizes both $\BoxPr$ and $\DiaPr$ and exhibits only the
  monotonicity rule (if the formula~$A \ImpPr B$ is valid, then so is
  $\MonPr{A} \ImpPr \MonPr{B}$)
\item the minimal lax logic~\SL with a modality~$\DiaPr$ that exhibits
  only the axiom~$\SA : A \AndPr \LaxPr{B} \ImpPr \LaxPr{(A \AndPr
    B)}$
\item the full lax logic~\LL that extends \SL with
  axioms~$\RA : A \ImpPr \LaxPr{A}$ and~$\JA : \LaxPr{\LaxPr{A}}
  \ImpPr \LaxPr{A}$
\item the minimal box logic~\CKBox with a modality~$\BoxPr$ that
  exhibits the necessitation rule (if $A$ is valid, then so is
  $\Box{A}$) and the distribution axiom~$\KA : \BoxPr{(A \ImpPr B)}
  \ImpPr \BoxPr{A} \ImpPr \BoxPr{B}$
\item the full box logic~\CSFourBox that extends \CKBox with
  axioms~$\TA : \BoxPr{A} \ImpPr A$
  and~$\FourA : \BoxPr{A} \ImpPr \BoxPr{\BoxPr{A}}$
\end{enumerate}

The logic~\CM can be found in a recent study of monotone logics by De
Groot~\cite[Definition 2.9]{Groot26}, where we take a single
monotone modality~$\MonPr$ instead of $\BoxPr$ and $\DiaPr$.
\CM is the logic of \emph{Modal Heyting algebras}~\cite[Section
  4]{Goldblatt11a} and the monotonicity rule in \CM corresponds to
functoriality in category theory.
For our purposes, \CM serves as a small toy logic with a
non\hyp{}trivial extension to \IPL that makes it easy to illustrate
the main ideas underlying modal cover semantics.
The logic~\SL (for ``S\hyp{}lax'' logic) is a minimal sublogic of \LL,
the latter of which has been studied
extensively~\cite{AlechinaMPR01,BentonBP98,FairtloughM97} and is well
known as the IML corresponding to strong monads~\cite{Moggi91}.
The axiom~\SA corresponds to strength of a functor, while the
axioms~\RA and~\JA correspond to the properties of a monad.
The logic~\CKBox is the smallest box\hyp{}only IML that underlies the
most widely studied box\hyp{}only IML~\CSFourBox.
\CKBox was given a dual\hyp{}context natural deduction system by
Kavvos~\cite{Kavvos17} by following the influential work of Pfenning
and Davies~\cite{PfenningD01} on \CSFourBox.
For our purposes, \CKBox and \CSFourBox serve as examples of IMLs
which can be modeled using modal cover semantics despite requiring a
special proof system that departs from the usual single\hyp{}context
systems used for the other logics.

We prove soundness for these logics by showing that modal cover
models determine equivalent algebraic models, and prove completeness
constructively in the style of Normalization by
Evaluation~\cite{CoquandD97,Coquand93,Coquand02} by constructing a
Henkin\hyp{}style canonical model.
Furthermore, we show that the completeness proofs can be readily
refined to give a normalization algorithm that normalizes proofs in
the respective natural deduction system of the logic, which yields as
corollaries the subformula property and logical consistency.
All theorems in this article have been manually formalized and
verified using the proof assistant~Agda\footnote{Source is available
  at \url{https://github.com/nachivpn/cover}}.

\section{Overview of Cover Semantics}\label{sec:overview}

In this section, we begin with a recap of cover semantics for \IPL
(\cref{sec:overview:ipl}) and give an overview of the trouble with
relational cover semantics (\cref{sec:overview:relational}).
We then illustrate our new semantics by defining modal cover
semantics for the logic~\CM (\cref{sec:overview:reg-cover}) and extend
this to the remaining IMLs in later sections.
The results in \cref{sec:overview:ipl,sec:overview:relational} are
well\hyp{}known and partially due to Goldblatt~\cite{Goldblatt11a}.

\subsection{Cover Semantics for \IPL}
\label{sec:overview:ipl}

The language of \IPL consists of formulas defined inductively by
propositional atoms ($p$,~$q$,~$r$, etc.), constants~$\TopPr$
and~$\BotPr$, and binary logical connectives~$\AndPr$,~$\OrPr$
and~$\ImpPr$.
As usual, the connectives~$\AndPr$ and~$\OrPr$ have higher operator
precedence than $\ImpPr$, and all binary connectives associate to the
right when they are nested.
\begin{align*}
  \Prop\ \ A,B  :=  p,q,r,\ldots\ |\ \TopPr\ |\ \BotPr\ |\ A \AndPr B\ |\ A \OrPr B\ |\ A \ImpPr B \qquad
  \Ctx\ \ \Gamma, \Delta  :=  \EmptyCtx\ |\ \ExtCtx{\Gamma}{A}
\end{align*}
The constants~$\TopPr$ and $\BotPr$ respectively denote universal
truth and falsity, and the connectives~$\AndPr$,~$\OrPr$
and~$\ImpPr$ respectively denote conjunction, disjunction and
implication.
A context~$\Gamma$ is a finite multiset of formulas
$A_1, A_2,..., A_n$, and $\EmptyCtx$ denotes the empty context.
A sequent\hyp{}style natural deduction proof system for \IPL is given
using the \emph{inference} rules defined in \cref{fig:system-ipl}.
A \emph{judgment}~$\Gamma \vdash A$ is an assertion that denotes
formula~$A$ has a proof under the assumption that all formulas in
context~$\Gamma$ have a proof.
A judgment~$\Gamma \vdash A$ \emph{holds}, written simply as
``$\Gamma \vdash A$'', when it can be derived using the inference
rules in \cref{fig:system-ipl}.

\begin{figure}[h]
  \begin{mathpar}
    \inferrule[Hyp]{%
      A \in \Gamma
    }{%
      \Gamma \vdash A
    }

    \inferrule[$\TopPr$\nbhyp{}Intro]{%
    }{%
      \Gamma \vdash \TopPr
    }%

    \inferrule[$\BotPr$\nbhyp{}Elim]{%
      \Gamma \vdash \BotPr
    }{%
      \Gamma \vdash A
    }

    \inferrule[$\AndPr$\nbhyp{}Intro]{%
      \Gamma \vdash A \\
      \Gamma \vdash B
    }{%
      \Gamma \vdash A \AndPr B
    }%

    \inferrule[$\AndPr$\nbhyp{}Elim\nbhyp{}1]{%
      \Gamma \vdash A \AndPr B
    }{%
      \Gamma \vdash A
    }%

    \inferrule[$\AndPr$\nbhyp{}Elim\nbhyp{}2]{%
      \Gamma \vdash A \AndPr B
    }{%
      \Gamma \vdash B
    }

    \inferrule[$\ImpPr$\nbhyp{}Intro]{
      \ExtCtx{\Gamma}{A} \vdash B
    }{%
      \Gamma \vdash A \ImpPr B
    }%

    \inferrule[$\ImpPr$\nbhyp{}Elim]{%
      \Gamma \vdash A \ImpPr B\\
      \Gamma \vdash A
    }{%
      \Gamma \vdash B
    }%

    \inferrule[$\OrPr$\nbhyp{}Intro\nbhyp{}1]{
      \Gamma\vdash A
    }{%
      \Gamma \vdash A \OrPr B
    }%

    \inferrule[$\OrPr$\nbhyp{}Intro\nbhyp{}2]{
      \Gamma\vdash B
    }{%
      \Gamma \vdash A \OrPr B
    }%

    \inferrule[$\OrPr$\nbhyp{}Elim]{%
      \Gamma \vdash A \OrPr B\\
      \ExtCtx{\Gamma}{A} \vdash C\\
      \ExtCtx{\Gamma}{B} \vdash C
    }{%
      \Gamma \vdash C
    }%
  \end{mathpar}
  \caption{Sequent\hyp{}style natural deduction for \IPL}
  \label{fig:system-ipl}
\end{figure}

In the cover semantics of \IPL, truth of formulas is defined using a
gadget called a \emph{cover system}.
A cover system~$C = (W,\Ri{}{},\Covp{}{})$ is a tuple consisting of a
set~$W$ of worlds, a reflexive\hyp{}transitive \emph{refinement}
relation~$\Ri{}{}$ on $W$, and a \emph{covering}
relation~$\Covp{}{}\ \subseteq W \times \Pow{(W)}$ subject to certain
conditions.
We write $\Ri{w}{w'}$ or $\RiOp{w'}{w}$, saying $w'$ \emph{refines}
$w$, to denote that the relation~$\Ri{}{}$ relates the world~$w$ to
the world~$w'$.
Similarly, we write $\Covp{w}{\alpha}$ or $\CovpOp{\alpha}{w}$, saying
$w$ is \emph{covered by} $\alpha$ or $\alpha$ \emph{covers} $w$ or
$\alpha$ is a \emph{cover of} $w$, to denote that the covering
relation~$\Covp{}{}$ relates the world~$w$ to a set $\alpha$
consisting of worlds.

We can define a refinement relation $\Rin{}{}\ \subseteq \Pow{(W)}
\times \Pow{(W)}$ on subsets of worlds using the refinement
relation~$\Ri{}{}$ on worlds as: $\Rin{\alpha}{\alpha'}$ if and only
if (iff) for all worlds~$v' \in \alpha'$ there exists a world~$v \in
\alpha$ such that $\Ri{v}{v'}$.
We write $\Rin{\alpha}{\alpha'}$ or $\RinOp{\alpha'}{\alpha}$, while
saying $\alpha'$ \emph{refines} $\alpha$.
The conditions on a cover system are:
\begin{itemize}
\item \emph{Refinement}: If $\RiOp{w'}{w} \Covp{}{} \alpha$, then there exists an $\alpha'$
  such that $\Covp{w'}{\alpha'} \RinOp{}{} \alpha$.
\item \emph{Inclusion}: If $\Covp{w}{\alpha}$, then
  $\Rin{\{w\}}{\alpha}$
\item \emph{Identity}: $\Covp{w}{\{w\}}$
\item \emph{Transitivity}: If $\Covp{w}{\alpha}$ and for all $v \in
  \alpha$ there exists an $\alpha_v$ such that $\Covp{v}{\alpha_v}$,
  then $\Covp{w}{\Union_{v \in \alpha} \alpha_v}$
\end{itemize}

We may intuitively understand a world~$w$ as a ``state of knowledge'',
the refinement~$\Ri{w}{w'}$ as increase in knowledge from $w$ to $w'$,
and a cover~$\alpha$ of a world~$w$ as defining a ``locality'' of
knowledge states capturing knowledge local to $w$.
Under this reading, the refinability condition ensures that local
knowledge improves with increase in ``current'' knowledge: if
$\Ri{w}{w'}$ and~$\Covp{w}{\alpha}$, then some cover of $w'$ refines
$\alpha$.
Similarly, the inclusion condition states that knowledge local to a
world~$w$ must refine current knowledge at~$w$.

A \emph{cover model}~$\Mod[M] = (C,V)$ of \IPL couples a cover
system~$C$ with a valuation function~$V$ mapping propositional atoms
to \emph{localized up\hyp{}sets} of $W$, i.e. a
function~$V : \Atom \to \Pow{(W)}$ satisfying the conditions:
\begin{itemize}
\item \emph{Upper set}: if $\Ri{w}{w'}$ and $w \in V(p)$, then $w' \in
  V(p)$
\item \emph{Localization}: if $\exists \alpha.\, \Covp{w}{\alpha} \subseteq V(p)$, then $w \in V(p)$
\end{itemize}
The valuation function~$V$ maps a propositional atom~$p$ to a subset
of worlds~$V(p) \subseteq W$ where $p$ is true.
The conditions respectively state that $V(p)$ must be an up\hyp{}set
of the preorder~$(W,\Ri{}{})$, and that if $p$ is ``locally true'' at
$w$, i.e. at all worlds in some cover $\alpha$ of $w$, then it must be
true at $w$.

Given a cover model~$\Mod[M]$, the truth of a formula~$A$ is given
using the \emph{satisfaction} relation~$\Vdash$, for an arbitrary
world~$w \in W$ underlying the model~$\Mod[M]$, by induction on the
formula as follows:
\begin{equation*}
  \begin{array}{l@{\;\Vdash\;} @{\;} l @{\;\text{iff}\;}c@{\;} l}
    \Mod[M],w & p & & w \in V(p) \\
    \Mod[M],w & \TopPr & & \text{true} \\
    \Mod[M],w & \BotPr & & \Covp{w}{\emptyset} \\
    \Mod[M],w & A \AndPr B & & \Mod[M],w \Vdash A\ \text{and}\  \Mod[M],w \Vdash B \\
    \Mod[M],w & A \OrPr B & & \exists \alpha.\, \Covp{w}{\alpha}\ \text{and}\ \forall v \in \alpha.\, \Mod[M],v \Vdash A\ \text{or}\ \Mod[M],v \Vdash B \\
    \Mod[M],w & A \ImpPr B & & \forall \RiOp{w'}{w}.\, \Mod[M],w' \Vdash A\ \text{implies}\ \Mod[M],w' \Vdash B\
  \end{array}
\end{equation*}

We extend the satisfaction relation to contexts and write
$\Mod[M],w \Vdash \Gamma$ to denote $\Mod[M], w \Vdash A_i$ for all
formulas~$A_i$ with $1 \leq i \leq n$ in $\Gamma = A_1,\ldots,A_n$.
We define the \emph{truth set}~$\truth{A}^{\Mod[M]}$ of a formula $A$
in some model~$\Mod[M]$ as the subset of worlds where $A$ is true, and
likewise extend this definition to contexts as follows:
$$\truth{A}^{\Mod[M]} = \{w \in W\ |\ \Mod[M],w \Vdash A\} \qquad
\truth{A_1,A_2,\ldots,A_n}^{\Mod[M]} = \truth{A_1}^{\Mod[M]} \cap
\truth{A_2}^{\Mod[M]}\ldots\cap \truth{A_n}^{\Mod[M]} $$
We sometimes omit the subscript~$\Mod[M]$ and write $\truth{A}$ or
$\truth{\Gamma}$ when it is evident from the context which model we
are working with.
We write $\Mod[\Gamma] \satisfies_{\Mod[M]} A$, saying $\Gamma$
\emph{entails} $A$ \emph{in model}~$\Mod[M]$, to denote that
$\truth{\Gamma}^{\Mod[M]} \subseteq \truth{A}^{\Mod[M]}$.
In other words, $\Gamma \satisfies_{\Mod[M]} A$ if and only if
$\Mod[M],w \Vdash \Gamma$ implies $\Mod[M],w \Vdash A$ for all
worlds~$w$ in model~$\Mod[M]$.
Furthermore, we write $\Gamma \satisfies A$, saying $\Gamma$
\emph{entails} $A$, to denote $\Mod[\Gamma] \satisfies_{\Mod[M]} A$
for all models~$\Mod[M]$.

To prove soundness for \IPL, we begin with the following lemma, which
observes that the conditions imposed on the truth of atoms are also
satisfied by truth sets of arbitrary formulas and contexts.
\begin{lemma}\label[lemma]{lem:ipl-truth-lupset}
  For any formula~$A$ and cover model~$\Mod[M]$ of \IPL, the truth set
  $\truth{A}^{\Mod[M]}$ is a localized up\hyp{}set. In other words, it
  satisfies the following properties.
  It follows that $\truth{\Gamma}^{\Mod[M]}$ is a localized
  up\hyp{}set for any context~$\Gamma$.
  \begin{itemize}
  \item Upper set: if $\Ri{w}{w'}$ and $w \in \truth{A}^{\Mod[M]}$, then
    $w' \in \truth{A}^{\Mod[M]}$
  \item Localization: if
    $\exists \alpha.\, \Covp{w}{\alpha} \subseteq \truth{A}^{\Mod[M]}$, then
    $w \in \truth{A}^{\Mod[M]}$
  \end{itemize}
\end{lemma}
\begin{proof}
  By induction on formula~$A$ and context~$\Gamma$. To prove that
  truth sets are up\hyp{}sets, we use the refinement condition for the
  cases of~$\bot$ and~$A \OrPr B$. To prove that truth sets are
  localizing, on the other hand, we use the inclusion condition for
  the case of $A \ImpPr B$ and the transitivity condition for the
  cases of~$\bot$ and~$A \OrPr B$. The remaining cases for both
  properties follow readily from the induction hypotheses.
\end{proof}

\begin{remark}
The statement of \cref{lem:ipl-truth-lupset} can be strengthened
further to show truth sets are \emph{hyper localized}, meaning
$\exists \alpha.\, \Covp{w}{\alpha} \subseteq \truth{A}^{\Mod[M]}$ if
and only if $w \in \truth{A}^{\Mod[M]}$.
This is because the identity condition ensures that every world is
covered by itself, meaning $\Covp{w}{\{w\}}$, thus forcing every
localized set to also be hyper localized.
The identity condition, however, is not always desirable and can be
replaced with a weaker condition such as Goldblatt's \emph{existence}
condition~\cite[Section 3]{Goldblatt11a}.
\end{remark}

\begin{proposition}[Soundness for \IPL]\label[prop]{prop:ipl-sound}
If \,$\Gamma \vdash A$ holds, then so does $\Gamma \satisfies A$.
\end{proposition}
\begin{proof}
  By induction on derivation of $\Gamma \vdash A$, using
  \cref{lem:ipl-truth-lupset} where needed.
  The interesting cases are:
  \begin{itemize}
  \item \emph{Rule~$\ImpPr$\nbhyp{}Intro}:
    We must show $\truth{\Gamma} \subseteq \truth{A \ImpPr B}$, which
    states that for all $w \in \truth{\Gamma}$ and all $\RiOp{w'}{w}$,
    if $w' \in \truth{A}$ then $w' \in \truth{B}$.
    By applying \cref{lem:ipl-truth-lupset} to $\Gamma$ we know
    $\truth{\Gamma}$ is an up\hyp{}set, and thus
    $w' \in \truth{\Gamma}$.
    Since $w' \in \truth{\Gamma}$ and $w' \in \truth{A}$, we know
    $w' \in \truth{\ExtCtx{\Gamma}{A}}$.
    Due to the induction hypothesis
    $\truth{\ExtCtx{\Gamma}{A}} \subseteq \truth{B}$ we thus have
    $w' \in \truth{B}$.
  \item \emph{Rule~$\BotPr$\nbhyp{}Elim}: We must show
    $\truth{\Gamma} \subseteq \truth{A}$.
    Suppose an arbitrary world~$w \in \truth{\Gamma}$.
    %
    From the induction
    hypothesis~$\truth{\Gamma} \subseteq \truth{\BotPr}$, we know
    $w \in \truth{\BotPr}$, which means
    $\Covp{w}{\emptyset} \subseteq \truth{A}$.
    By applying \cref{lem:ipl-truth-lupset} to $A$, we know
    $\truth{A}$ satisfies localization, and thus it must be case that
    $w \in \truth{A}$.

  \item \emph{Rule~$\OrPr$\nbhyp{}Elim}: We must show
    $\truth{\Gamma} \subseteq \truth{C}$ from the induction
    hypotheses~$\truth{\Gamma} \subseteq \truth{A \OrPr B}$ (IH.1),
    $\truth{\ExtCtx{\Gamma}{A}} \subseteq \truth{C}$ (IH.2) and
    $\truth{\ExtCtx{\Gamma}{B}} \subseteq \truth{C}$ (IH.3).
    Suppose an arbitrary world~$w \in \truth{\Gamma}$.
    From IH.1, we known~$w \in \truth{A \OrPr B}$, which means all
    members of some cover~$\CovpOp{\alpha}{w}$ are either in
    $\truth{A}$ or $\truth{B}$.
    Consider an arbitrary~$v \in \alpha$.
    The inclusion condition ensures that $v$ refines $w$.
    By applying \cref{lem:ipl-truth-lupset} to $\Gamma$, we know
    $\truth{\Gamma}$ is an up\hyp{}set, which means
    $v \in \truth{\Gamma}$.
    If $v \in \truth{A}$, then $v \in \truth{\ExtCtx{\Gamma}{A}}$ and
    thus $v \in \truth{C}$ by IH.2.
    Otherwise $v \in \truth{B}$, then
    $v \in \truth{\ExtCtx{\Gamma}{B}}$ and thus $v \in \truth{C}$ by
    IH.3.
    As a result, for any~$v \in \alpha$ we have $v \in \truth{C}$,
    meaning $\Covp{w}{\alpha} \subseteq \truth{C}$.
    By applying \cref{lem:ipl-truth-lupset} to $C$, we know
    $\truth{C}$ satisfies localization, and thus $w \in \truth{C}$ as
    desired.
  \end{itemize}
\end{proof}

We now turn our attention to proving completeness.
Following the standard practice, we will achieve this by constructing
a \emph{canonical} model~$\Mod[N]$ that equates entailment in the
model~$\Gamma \satisfies_{\Mod[N]} A$ with derivability of the
corresponding judgment~$\Gamma \vdash A$.
We begin with a few definitions for this purpose.

\begin{definition}
  Given a formula~$A$, we define the set~$\AntTh{(A)}$, called the
  \emph{antecedents} of $A$, as the set of contexts that~$A$ can be
  proved under the assumption of, i.e.
  $\AntTh{(A)} = \{\Gamma \in \Ctx\ |\ \Gamma \vdash A\}$.
\end{definition}

\begin{definition}\label[def]{def:cov-sys-ipl}
  Define a cover
  system~$C_{\text{\IPL}} = \tuple{\Ctx,\subseteq,\CovpIPL{}{}}$ by
  taking the set~$\Ctx$ of contexts for worlds~$W$, the context
  inclusion relation~$\subseteq$ for the preorder relation~$\Ri{}{}$,
  and the below inductively defined
  relation~$\CovpIPL{}{}\ \subseteq \Ctx \times \Pow{(\Ctx)}$ for the
  covering relation~$\Covp{}{}$.
  The relation~$\CovpIPL{}{}$ can be verified to satisfy the
  refinability, inclusion, identity and transitivity conditions by
  induction on its definition.
  \begin{mathpar}
    \inferrule[]{%
    }{%
      \CovpIPL{\Gamma}{\{\Gamma\}}
    }

    \inferrule[]{%
      \Gamma \in \AntTh{(\BotPr)}}{%
      \CovpIPL{\Gamma}{\emptyset}
    }

    \inferrule[]{%
      \Gamma \in \AntTh{(A \OrPr B)}\\
      \CovpIPL{\ExtCtx{\Gamma}{A}}{\alpha_1} \\
      \CovpIPL{\ExtCtx{\Gamma}{B}}{\alpha_2} \\
    }{%
      \CovpIPL{\Gamma}{\alpha_1 \union \alpha_2}
    }
  \end{mathpar}
\end{definition}

\begin{example}
As an example, observe that the judgment~$A \OrPr B \vdash A \OrPr B$
is derivable, which means $A \OrPr B \in \AntTh{(A \OrPr B)}$.
With this observation, we can
show~$\CovpIPL{A \OrPr B}{\{(\ExtCtx{A \OrPr B}{A}),(\ExtCtx{A \OrPr
    B}{B})\}}$, stating that the singleton context~$A \OrPr B$ is
covered by the set of two
contexts~$\{(\ExtCtx{A \OrPr B}{A}),(\ExtCtx{A \OrPr B}{B})\}$ as
below:
\begin{mathpar}
  \inferrule[]{%
    \inferrule[]{}{A \OrPr B  \in \AntTh{(A \OrPr B)}} \\
    \CovpIPL{\ExtCtx{A \OrPr B}{A}}{\{(\ExtCtx{A \OrPr B}{A})\}} \\
    \CovpIPL{\ExtCtx{A \OrPr B}{B}}{\{(\ExtCtx{A \OrPr B}{B})\}} \\
  }{%
    \CovpIPL{A \OrPr B}{\{(\ExtCtx{A \OrPr B}{A}),(\ExtCtx{A \OrPr B}{B})\}}
  }
\end{mathpar}

\end{example}

\begin{lemma}[Truth Lemma]\label[lem]{lem:ipl-truth-lemma}
  The tuple~$\Mod[N] = \tuple{C_{\text{\IPL}},\AntTh}$ is a cover
  model of \IPL with the characteristic property that, for every
  formula~$A$ and context~$\Gamma$, we have
  $\Gamma \in \truth{A}^{\Mod[N]}$ if and only if\,
  $\Gamma \in \AntTh{(A)}$.
\end{lemma}
\begin{proof}
  We first establish for an arbitrary formula~$A$ that the
  set~$\AntTh{(A)}$ is a localized up\hyp{}set of the cover
  system~$C_{\text{\IPL}}$, by induction on the definition of the
  covering relation~$\CovpIPL{}{}$.
  This ensures that~$\AntTh$ can be used for the valuation of
  propositional atoms and that the tuple~$\Mod[N]$ is indeed a cover
  model of \IPL.

  We then show the ``characteristic'' property of $\Mod[N]$ by
  induction on the formula~$A$.
  The interesting cases are those of the connectives~$\BotPr$ and
  $\OrPr$, both of which use the fact that antecedents are localized
  up\hyp{}sets.
  \begin{itemize}
  \item \emph{Case~$A = \BotPr$}: From left to right, if
    $\Gamma \in \truth{\BotPr}^{\Mod[N]}$ then
    $\CovpIPL{\Gamma}{\emptyset} \subseteq \AntTh{(\BotPr)}$, which
    when combined with the fact that $\AntTh{(\BotPr)}$ is a localized
    up-set, means $\Gamma \in \AntTh{(\BotPr)}$.
    In the converse direction, if $\Gamma \in \AntTh{(\BotPr)}$ then
    $\CovpIPL{\Gamma}{\emptyset}$ by definition of the covering
    relation~$\CovpIPL{}{}$ and thus $\Mod[N],\Gamma \Vdash \BotPr$,
    or equivalently $\Gamma \in \truth{\BotPr}^{\Mod[N]}$.

  \item \emph{Case~$A = A_1 \OrPr A_2$}: From left to right, if
    $\Gamma \in \truth{A_1 \OrPr A_2}^{\Mod[N]}$ we learn that there
    exists some~$\alpha$ such
    that~$\CovpIPL{\Gamma}{\alpha} \subseteq \truth{A_1}^{\Mod[N]}
    \union \truth{A_2}^{\Mod[N]}$ by unpacking definition of
    $\truth{A_1 \OrPr A_2}^{\Mod[N]}$.
    Using the induction hypotheses we can show that
    $\truth{A_i}^{\Mod[N]} \subseteq \AntTh{(A_i)} \subseteq
    \AntTh{(A_1 \OrPr A_2)}$ for $i \in \{1,2\}$, which means
    $\CovpIPL{\Gamma}{\alpha} \subseteq \AntTh{(A_1 \OrPr A_2)}$
    In light of the fact that $\AntTh{(A_1 \OrPr A_2)}$ is a localized
    up-set, this means $\Gamma \in \AntTh{(A_1 \OrPr A_2)}$.

    In the converse direction, if $\Gamma \in \AntTh{(A_1 \OrPr A_2)}$
    then we can
    show~$\CovpIPL{\Gamma}{\{(\ExtCtx{\Gamma}{A_1}),(\ExtCtx{\Gamma}{A_2})\}}$.
    Since $(\ExtCtx{\Gamma}{A_i}) \in \AntTh{(A_i)}$, we have
    $(\ExtCtx{\Gamma}{A_i}) \in \truth{A_i}^{\Mod[N]}$ by virtue of
    the respective induction hypothesis.
    Altogether, we have that all members of some cover of $\Gamma$ are
    in $\truth{A_1}^{\Mod[N]}$ or $\truth{A_2}^{\Mod[N]}$, which means
    $\Gamma \in \truth{A_1 \OrPr A_2}^{\Mod[N]}$.

  \end{itemize}
\end{proof}

\begin{theorem}[Completeness for \IPL]\label[thm]{thm:ipl-complete}
If \,$\Gamma \satisfies A$, then $\Gamma \vdash A$.
\end{theorem}
\begin{proof}
  We first show $\Gamma \in \AntTh{(\bigwedge \Gamma)}$, equivalently
  $\Gamma \vdash \bigwedge \Gamma$, using the inference rules for
  \IPL by induction on the context~$\Gamma$.
  By applying \Cref{lem:ipl-truth-lemma} (from right to left of the
  bi\hyp{}implication) we infer
  $\Gamma \in \truth{\bigwedge \Gamma}^{\Mod[N]}$.
  Moreover, we observe
  $\truth{\bigwedge \Gamma}^{\Mod[N]} = \truth{\Gamma}^{\Mod[N]}$ by
  appealing to the definition of truth sets and further infer
  $\Gamma \in \truth{\Gamma}^{\Mod[N]}$.

  Given \,$\Gamma \satisfies A$, we have
  $\truth{\Gamma}^{\Mod[N]} \subseteq \truth{A}^{\Mod[N]}$ since
  $\Mod[N]$ is a model of \IPL.
  This means we have $\Gamma \in \truth{A}^{\Mod[N]}$, and by applying
  \Cref{lem:ipl-truth-lemma} once again (from left to right), we
  conclude $\Gamma \in \AntTh{(A)}$ and thus $\Gamma \vdash A$.
\end{proof}

\subsection{Relational Cover Semantics for IMLs}
\label{sec:overview:relational}

The language of the intuitionistic modal logic~\CM extends that of
\IPL with a unary connective~$\MonPr$, and its natural deduction proof
system extends that of \IPL with a rule~\CM/$\MonPr$\nbhyp{}Mon (for
``monotonicity'').
\begin{align*}
  \Prop\ \ A,B  & :=  \ldots\ |\ \MonPr{A} \qquad
  \inferrule[\CM/$\MonPr$\nbhyp{}Mon]{%
      \Gamma  \vdash \MonPr{A} \\
      A \vdash B
    }{%
      \Gamma \vdash \MonPr{B}
    }
\end{align*}

A \emph{relational cover system}~$\tuple{C,\Rm}$ extends the
definition of a cover system~$C = \tuple{W,\Ri{}{},\Covp{}{}}$ with an
\emph{accessibility} relation~$\Rm$.
The relation $\Rm$ is a binary relation on worlds subject to the modal
refinability and localization conditions stated below.
We write $w \Rm v$ or $v \RmOp w$, and say $w$ \emph{can access} $v$
or $v$ is \emph{accessible from} $w$, to denote that the world~$w$ is
related to world~$v$ via the relation~$\Rm$.
To state the modal conditions, we define an operator~$\rOp$ on subsets
of $W$, for a given $X \subseteq W$, as:
$$\rOp{X} = \{w \in W\ |\ \exists x \in X.\, w \Rm x \}$$
The set~$\rOp{X}$ identifies all worlds that can access some
world in $X$.
The modal conditions are:
\begin{itemize}
\item \emph{Modal Refinement}: If $\RiOp{w'}{w} \Rm v$, then there
  exists a $v'$ such that $w' \Rm \RiOp{v'}{v}$
\item \emph{Modal Localization}: If
  $\Covp{w}{\alpha} \subseteq \rOp{X}$, then there exists a $v$ and
  $\alpha_v$ such that $w \Rm \Covp{v}{\alpha_v} \subseteq X$.
\end{itemize}

A \emph{relational cover model}~$\Mod[M] = \triple{C}{R}{V}$ of \CM
couples a relational cover system~$\tuple{C,R}$ with a valuation
function~$V$ mapping propositional atoms to localized up\hyp{}sets of
$W$---as before with \IPL.
The truth of \CM formulas is defined by extending the satisfaction
relation for \IPL to modal formulas as follows:
\begin{equation*}
  \begin{array}{l@{\;\Vdash\;} @{\;} l @{\;\text{iff}\;}c@{\;} l}
    \Mod[M],w & \MonPr{A} & & \exists v.\, w \Rm v\  \text{and}\ \Mod[M],v \Vdash A
  \end{array}
\end{equation*}
To prove soundness for \CM, we re\hyp{}establish
\cref{lem:ipl-truth-lupset} by showing that truth sets for formulas in
\CM are indeed localized up\hyp{}sets, using the modal conditions for
the case of modal formulas~$\MonPr{A}$.
We then prove soundness for \CM by induction on the derivations of
judgments, as in the proof of \cref{prop:ipl-sound}.
The trouble, however, lies in proving completeness.
Constructing a canonical relational cover model of \CM
requires us to extend the cover system in \cref{def:cov-sys-ipl} with
a relation on contexts.
A natural candidate for a relation in the canonical model~$\Mod[N]$
would be the relation~${\RCM} \subseteq \Ctx \times \Ctx$ defined as
follows:
$$\Gamma \RCM \Delta\
\text{iff there exists a formula}~A\ \text{s.t.}\ \Gamma \in
\AntTh{{(\MonPr{A})}}\ \text{and}\ \Delta = \{A\}$$
The relation~$\RCM$ has the essential character of equating entailment
in the model with provability.
We can show that the truth set~$\truth{\MonPr{B}}$ determined by the
relation~$\RCM$ is in fact equivalent to the
set~$\AntTh{(\MonPr{B})}$.
However, $\RCM$ crucially fails to satisfy modal localization,
blocking us from using it to construct a relational cover
model---inhibiting a proof of the truth lemma
(\cref{lem:ipl-truth-lemma}) used to show completeness.
Goldblatt~\cite[Section 8]{Goldblatt11a} encounters a similar
roadblock in an attempt to construct a Henkin\hyp{}style model of
\LL.

At first sight, it may appear as though the modal localization
condition is at fault.
However, the modal localization condition on relational cover systems
simply states what is required to show that truth sets of modal
formulas satisfy the localization property.
Based on observations from the model constructions to follow in this
article, our experience suggests that the interpretation of modal
formulas in relational cover semantics is itself somewhat restrictive.
An accessibility relation forces us to choose \emph{exactly one}
possible future world to witness the truth of a modal formula, while
many models necessitate a \emph{collection} of possible future worlds.
In the upcoming sections, we develop a conservative extension of
relational cover semantics by replacing the accessibility
relation~$\Rm$ with a modal covering relation~$\Covm{}{}$ to alleviate
this restriction.

\subsection{Modal Cover Semantics for IMLs}
\label{sec:overview:reg-cover}

A \emph{modal cover system}~$\tuple{C,\Covm{}{}}$ extends the
definition of a cover system~$C = \tuple{W,\Ri{}{},\Covp{}{}}$ with a
\emph{modal covering} relation~$\Covm{}{}$ subject to the modal
refinability and localization conditions stated below.
We may intuitively understand a modal cover~$\beta$ of a world~$w$,
written $\Covm{w}{\beta}$ or $\CovmOp{\beta}{w}$, as defining a
``speculation'' about possible ``future'' states of knowledge based on
``current'' knowledge at $w$.
To state the modal conditions, we define two operators~$\jOp{}$
and~$\mOp{}$ on subsets of $W$, for a given~$X \subseteq W$, as:
\begin{align*}
  \jOp{X} = \{ w \in W\ |\ \exists \alpha.\, \Covp{w}{\alpha} \subseteq X\} \qquad
  \mOp{X} = \{ w \in W\ |\ \exists \alpha.\, \Covm{w}{\alpha} \subseteq X\}
\end{align*}
The set~$\jOp{(X)}$ identifies all worlds~$w$ locally covered by some
subset~$\alpha$ of $X$, while the set~$\mOp{X}$ identifies all
worlds~$w$ modally covered by some subset~$\alpha$ of $X$.
The modal conditions are:
\begin{itemize}
\item \emph{Modal Refinement}: If $\RiOp{w'}{w} \Covm{}{} \alpha$,
  then there exists an $\alpha'$ such that
  $\Covm{w'}{\alpha'} \RinOp{}{} \alpha$.
\item \emph{Modal Localization}: If
  $\Covp{w}{\alpha} \subseteq \mOp{X}$, then there exists a
  $\beta$ such that $\Covm{w}{\beta} \subseteq \jOp{X}$
\end{itemize}
The modal conditions on modal cover systems generalize those on
relational cover systems by respectively replacing the accessibility
relation~$\Rm$ and operator~$\rOp$ with the modal covering
relation~$\Covm{}{}$ and operator~$\mOp$.
As before, the modal conditions ensure that truth sets of modal
formulas are localized up\hyp{}sets.

%

A \emph{modal cover model}~$\Mod[M] = \triple{C}{\Covm{}{}}{V}$ of
\CM couples a modal cover system~$\tuple{C,\Covm{}{}}$ with a
valuation function~$V$ that maps atoms to localized up\hyp{}sets of
$W$, which is a function $V : \Atom \to \Pow{(W)}$ satisfying the
upper set and localization conditions imposed on a cover model of
\IPL.
Observe that there is no modal counterpart to the localization
condition on $V$ concerning the local covering relation~$\Covp{}{}$.
Intuitively, this is because we cannot expect a formula~$A$ that is
``speculatively true'' at a world~$w$ to become true at $w$.
The truth of \CM formulas is defined for a given modal cover
model~$\Mod[M]$ of \CM by extending the definition of the satisfaction
relation for \IPL to modal formulas as follows:
\begin{equation*}
  \begin{array}{l@{\;\Vdash\;} @{\;} l @{\;\text{iff}\;}c@{\;} l}
    \Mod[M],w & \MonPr{A} & & \exists \beta.\, \Covm{w}{\beta}\ \text{and}\ \forall v \in \beta.\, \Mod[M],v \Vdash A
  \end{array}
\end{equation*}
This definition states that a modal formula~$\MonPr{A}$ is true at a
world $w$ iff $A$ is true at all members~$v$ of some modal
cover~$\beta$ of $w$.
In contrast, recollect that the relational approach in the previous
subsection requires $A$ to be true at some world~$v$ accessible from
$w$.
This means we can recover the relational semantics for \CM from the
modal cover semantics for \CM by simply restricting modal covers to be
singletons.

\begin{lemma}\label[lemma]{lem:cm-truth-lupset}
  For any modal cover model~$\Mod[M]$ of \CM, the truth sets
  $\truth{A}^{\Mod[M]}$ and $\truth{\Gamma}^{\Mod[M]}$ are localized
  up\hyp{}sets.
\end{lemma}
\begin{proof}
  By repeating the induction in \cref{lem:ipl-truth-lupset}, using the
  modal conditions for modal formulas.
\end{proof}

\begin{proposition}[Soundness for \CM]\label[prop]{prop:cm-sound}
If \,$\Gamma \vdash A$, then $\Gamma \satisfies A$.
\end{proposition}
\begin{proof}
  By repeating the induction in proof of \cref{prop:ipl-sound}, now using
  \cref{lem:cm-truth-lupset}.
  The interesting case is that of \emph{Rule
    \CM/$\MonPr{}$\nbhyp{}Mon}:
  We must show $\truth{\Gamma} \subseteq \truth{\MonPr{B}}$ from the
  induction hypotheses~$\truth{\Gamma} \subseteq \truth{\MonPr{A}}$
  (IH.1) and $\truth{A} \subseteq \truth{B}$ (IH.2).
  If some~$w \in \truth{\Gamma}$, it follows from IH.1 that
  $w \in \truth{\MonPr{A}}$, which means for some~$\beta$,
  $\Covm{w}{\beta}$ and $\beta \subseteq \truth{A}$.
  It follows from IH.2 that $\beta \subseteq \truth{B}$, which means
  we also have $w \in \truth{\MonPr{B}}$ as desired.
\end{proof}

To prove completeness for \CM, let us define a modal cover
system~$C_{\text{\CM}} = \tuple{C_{\text{\IPL}},\CovmCM{}{}}$ coupling
the cover system~$C_{\text{\IPL}}$ (reproducing \cref{def:cov-sys-ipl}
in the language of \CM) with the modal covering
relation~$\CovmCM{}{} \subseteq \Ctx \times \Pow{(\Ctx)}$ defined
inductively below.
The relation~$\CovmCM{}{}$ can be verified to satisfy the modal
refinability and localization conditions by induction on its
definition.
\begin{mathpar}
  \inferrule[]{%
    \Gamma \in \AntTh{(\MonPr{A})}
  }{%
    \CovmCM{\Gamma}{\{A\}}
  }

  \inferrule[]{%
    \Gamma \in \AntTh{(\BotPr)}}{%
    \CovmCM{\Gamma}{\emptyset}
  }

  \inferrule[]{%
    \Gamma \in \AntTh{(A \OrPr B)}\\
    \CovmCM{\ExtCtx{\Gamma}{A}}{\alpha_1} \\
    \CovmCM{\ExtCtx{\Gamma}{B}}{\alpha_2} \\
  }{%
    \CovmCM{\Gamma}{\alpha_1 \union \alpha_2}
  }
\end{mathpar}
Observe that there is an overlap in the definitions of the modal
($\CovmCM{}{}$) and local ($\CovpIPL{}{}$) covering relations used to
define $C_{\text{\CM}}$.
This overlap allows us to show that the modal localization condition
holds for $C_{\text{\CM}}$.
\begin{lemma}[Truth Lemma]\label[lem]{lem:cm-truth-lemma}
  The tuple~$\Mod[N] = \tuple{C_{\text{\CM}},\AntTh}$ is a modal
  cover model of \CM with the characteristic property that, for every
  formula~$A$ in \CM, we have $\Gamma \in \truth{A}^{\Mod[N]}$ if and
  only if\, $\Gamma \in \AntTh{(A)}$.
\end{lemma}
\begin{proof}
  By repeating the induction on formulas in
  \cref{lem:ipl-truth-lemma}.
  The interesting case is that of modal formulas~$\MonPr{B}$.
  From left to right: if $\Gamma \in \truth{\MonPr{B}}$ then for
  some~$\beta$, $\CovmCM{\Gamma}{\beta}$ and
  $\beta \subseteq \truth{B}$.
  By applying the induction hypothesis on $B$, we know
  $\truth{B} \subseteq \AntTh{(B)}$, which means
  $\beta \subseteq \AntTh{(B)}$.
  By induction on the relation~$\CovmCM{}{}$, we can show that
  $\CovmCM{\Gamma}{\beta} \subseteq \AntTh{(B)}$ implies
  $\Gamma \in \AntTh{(\MonPr{B})}$ as desired.
  From right to left: if $\Gamma \in \AntTh{(\MonPr{B})}$, then
  $\CovmCM{\Gamma}{\{B\}}$.
  By applying the IH once again on $B$, we know
  $\AntTh{(B)} \subseteq \truth{B}$, which means we have
  $\{B\} \subseteq \truth{B}$ since $\{B\} \subseteq \AntTh{(B)}$.
  Altogether $\CovmCM{\Gamma}{\{B\}} \subseteq \truth{B}$ and thus
  $\Gamma \in \truth{\MonPr{B}}$.  
\end{proof}

\begin{theorem}[Completeness for \CM]\label[thm]{thm:cm-complete}
If \,$\Gamma \satisfies A$, then $\Gamma \vdash A$.
\end{theorem}
\begin{proof}
  By repeating the argument in \cref{thm:ipl-complete}, now using
  \cref{lem:cm-truth-lemma}.
\end{proof}

\section{Semantic Analysis of Modal Cover Models}

The operators~$\jOp$ and $\mOp$ defined in the previous section
possess a number of general algebraic properties that make it possible
to avoid repetition in the proofs of soundness and completeness for
various IMLs with respect to their modal cover semantics.
We will identify these properties in this section.

Given a cover system~$C = \tuple{W,\Ri{}{},\Covp{}{}}$, recollect that
an \emph{upper set} or \emph{up\hyp{}set} is an ``upwards closed''
subset~$X \subseteq W$ with the property that if $\Ri{w}{w'}$ and $w
\in X$, then $w' \in X$.
Up\hyp{}sets can be characterized using an operator~$\uOp$ defined on
subsets of $W$
as~$\uOp{X} = \{ w \in W\ |\ \exists x \in X.\, \Ri{x}{w} \}$.
The set~$\uOp{X}$ identifies all worlds~$w$ that refine some world~$x$
in $X$, and $X$ is an up\hyp{}set if and only if $\uOp{X} = X$.
Recollect similarly that a \emph{localized} or \emph{localizing} set,
is a subset~$X \subseteq W$ with the property that if
$\exists \alpha.\, \Covp{w}{\alpha} \subseteq X$, then $w \in X$.
Localized sets can be characterized using the
operator~$\jOp{X} = \{ w \in W\ |\ \exists \alpha.\, \Covp{w}{\alpha}
\subseteq X\}$ from earlier.
The set~$\jOp{X}$ identifies all worlds~$w$ locally covered by some
subset~$\alpha$ of $X$, and a set~$X$ is a localized set if and only
if $\jOp{X} \subseteq X$.
We refer to an up\hyp{}set as a \emph{localized up\hyp{}set} if it is
also localized.
We write~$\Up{(W)}$ to denote the collection of all up\hyp{}sets
and~$\LUp{(W)}$ to denote the collection of all localized
up\hyp{}sets.

\begin{proposition}\label[prop]{prop:jop-properties}
  The operator~$\jOp : \Pow{(W)} \to \Pow{(W)}$ exhibits the following properties:
\begin{enumerate}
\item $\jOp$ is a closure operator on subsets of $W$, i.e. for all
  $X, Y \in \Pow{(W)}$
  $$X \subseteq Y \text{ implies }  \jOp{X} \subseteq \jOp{Y} \qquad X \subseteq \jOp{X} \qquad \jOp{\jOp{X}} \subseteq \jOp{X}$$
\item $\jOp$ defines a nucleus on up\hyp{}sets of $W$, i.e. we have a
  closure operator~$\jOp : \Up{(W)} \to \Up{(W)}$ that additionally
  satifies the inequality~$X \cap \jOp{Y}\subseteq \jOp{(X \cap Y)}$
  for all $X,Y \in \Up{(W)}$
\item $\jOp$ defines an nucleus on localized up\hyp{}sets of $W$,
  i.e. we have a nucleus~$\jOp : \LUp{(W)} \to \LUp{(W)}$.
\end{enumerate}
\end{proposition}
\begin{proof}
  In property~(i), monotonicity of $\jOp$ follows from its definition,
  and the two inequalities follow respectively from the identity and
  transitivity conditions.
  For property~(ii), we use the refinement condition to show
  that~$\jOp{X}$ must be an up\hyp{}set if $X$ is, and the inclusion
  condition to show the additional
  inequality~$X \cap \jOp{Y}\subseteq \jOp{(X \cap Y)}$ by leveraging
  the fact that $X$ is an up-set.
  For property~(iii), we use property~(ii) to observe
  $\jOp{\jOp{X}} \subseteq \jOp{X}$ for any up\hyp{}set~$X$, and thus
  $\jOp{X}$ is a localized up\hyp{}set if $X$ is.
\end{proof}

Recall that a Heyting algebra~$H = (U,\leq,\times,+,1,0,\Rightarrow)$
is a lattice~$(U,\leq,\times,+)$ consisting of a partial order~$\leq$
on a carrier set~$U$ with meet ($\times$) and join ($+$) operations,
accompanied by a maximal element~$1$, a minimal element~$0$, and an
operation~$\Rightarrow$ on $U$ such that $c \leq a \Rightarrow b$ if
and only if $c \times a \leq b$, for all elements~$a,b,c \in U$.
Further recall that an \emph{algebraic model}~$\Mod[A] = \tuple{H,V}$
of \IPL consists of a Heyting algebra~$H$ and valuation
function~$V : \Atom \to U$ mapping atoms to elements of the set~$U$.
For any algebraic model~$\Mod[A]$ of \IPL, we can extend the valuation
of atoms to give an \emph{interpretation} of
formulas~$\eval{-}^{\Mod[A]} : \Prop \to U$ and
contexts~$\eval{-}^{\Mod[A]} : \Ctx \to U$ in the carrier set~$U$ of
the underlying Heyting algebra~$H$ as follows:
\vspace{-\abovedisplayskip}
\begin{center}
\begin{minipage}{.2\linewidth}
\centering
\begin{equation*}
\begin{array}{>{\evallbracket\,}l@{\,\evalrbracket^{\Mod[A]}}l @{\;}c@{\;} l}
  p           && = & V(p) \\
  \TopPr     && = & 1 \\
  \BotPr    && = & 0
\end{array}
\end{equation*}
\end{minipage}%
\begin{minipage}{.4\linewidth}
\centering
\begin{equation*}
\begin{array}{>{\evallbracket\,}l@{\,\evalrbracket^{\Mod[A]}}l @{\;}c@{\;} l}
  A \AndPr B && = & \eval{A}^{\Mod[A]} \times \eval{B}^{\Mod[A]}\\
  A \OrPr B && = & \eval{A}^{\Mod[A]} + \eval{B}^{\Mod[A]}\\
  A \ImpPr B  && = & \eval{A}^{\Mod[A]} \Rightarrow \eval{B}^{\Mod[A]}
  \end{array}
\end{equation*}
\end{minipage}
\begin{minipage}{.3\linewidth}
\centering
\begin{equation*}
\begin{array}{>{\evallbracket\,}l@{\,\evalrbracket^{\Mod[A]}}l @{\;}c@{\;} l}
  \EmptyCtx && = & 1\\
  \Gamma, A && = & \eval{\Gamma}^{\Mod[A]} \times \eval{A}^{\Mod[A]}
\end{array}
\end{equation*}
\end{minipage}%
\end{center}
Moreover, it is well known that these functions are both sound and
complete, meaning a judgment~$\Gamma \vdash A$ is derivable in a proof
system for \IPL, if and only if, $\eval{\Gamma}^{\Mod[A]} \leq
\eval{A}^{\Mod[A]}$ holds for all algebraic models~$\Mod[A]$ of \IPL.

\begin{proposition}\label[prop]{prop:cs-to-ha}
  Every cover system~$C=\tuple{W,\Ri{}{},\Covp{}{}}$ determines a
  Heyting algebra~$\Psh[C]$ defined by taking:
\begin{itemize}
\item $\LUp(W)$ as the carrier ordered by set inclusion~$\subseteq$
\item $X \cap Y$ as the meet of localized up\hyp{}sets~$X$ and~$Y$
\item $\jOp{(X \cup Y)}$ as the join of  localized up\hyp{}sets~$X$ and~$Y$
\item $W$ as the maximal element and $\jOp{(\emptyset)}$ as the minimal element
\item $X \Rightarrow Y = \{ w\ | \uOp{\{w\}} \cap X \subseteq Y \}$ as
  the exponent of localized up\hyp{}sets~$X$ and~$Y$
\end{itemize}
\end{proposition}
\begin{proof}
  Using the relevant definitions and properties of the operator~$\jOp$
  in \cref{prop:jop-properties}.
\end{proof}

\begin{proposition}\label[prop]{prop:cm-to-am}
  Every cover model~$\Mod[M]=\tuple{C,V}$ of \IPL determines an
  equivalent algebraic model~$\Psh[\Mod[M]]=\tuple{\Psh[C],V}$ of \IPL
  such that $\truth{A}^{\Mod[M]} = \eval{A}^{\Psh[\Mod[M]]}$ and
  $\truth{\Gamma}^{\Mod[M]} = \eval{\Gamma}^{\Psh[\Mod[M]]}$ for all
  formulas~$A$ and contexts~$\Gamma$ in \IPL.
\end{proposition}
\begin{proof}
  It follows from \cref{prop:cs-to-ha} and the conditions
  on the function~$V$ (in a cover model of \IPL) that $\Psh[\Mod[M]]$
  is indeed an algebraic model of \IPL.
  This means we obtain an interpretation of a formula~$A$ as a
  localized up\hyp{}set~$\eval{A}^{\Psh[\Mod[M]]}$.
  This interpretation can be given explicitly by induction on $A$, and
  extended to a context~$\Gamma$, as follows:
\vspace{-\abovedisplayskip}
\begin{center}
\begin{minipage}{.2\linewidth}
\centering
\begin{equation*}
\begin{array}{>{\evallbracket\,}l@{\,\evalrbracket^{\Psh[\Mod[M]]}}l @{\;}c@{\;} l}
  p           && = & V(p) \\
  \TopPr     && = & W \\
  \BotPr    && = & \jOp{(\emptyset)}
\end{array}
\end{equation*}
\end{minipage}%
\begin{minipage}{.4\linewidth}
\centering
\begin{equation*}
\begin{array}{>{\evallbracket\,}l@{\,\evalrbracket^{\Psh[\Mod[M]]}}l @{\;}c@{\;} l}
  A \AndPr B && = & \eval{A}^{\Psh[\Mod[M]]} \cap \eval{B}^{\Psh[\Mod[M]]}\\
  A \OrPr B && = & \jOp{(\eval{A}^{\Psh[\Mod[M]]} \cup \eval{B}^{\Psh[\Mod[M]]})}\\
  A \ImpPr B  && = & \eval{A}^{\Psh[\Mod[M]]} \Rightarrow \eval{B}^{\Psh[\Mod[M]]}
  \end{array}
\end{equation*}
\end{minipage}
\begin{minipage}{.3\linewidth}
\centering
\begin{equation*}
\begin{array}{>{\evallbracket\,}l@{\,\evalrbracket^{\Psh[\Mod[M]]}}l @{\;}c@{\;} l}
  \EmptyCtx && = & W\\
  \Gamma, A && = & \eval{\Gamma}^{\Psh[\Mod[M]]} \cap \eval{A}^{\Psh[\Mod[M]]}
\end{array}
\end{equation*}
\end{minipage}%
\end{center}
It can readily observed by induction that for any $A$ and $\Gamma$ in
\IPL, $\truth{A}^{\Mod[M]} = \eval{A}^{\Psh[\Mod[M]]}$ and
$\truth{\Gamma}^{\Mod[M]} = \eval{\Gamma}^{\Psh[\Mod[M]]}$.
\end{proof}

%


We now turn our attention to the modal operator~$\mOp$.
Given a modal cover system~$\tuple{C,\Covm{}{}}$, recollect that
the operator~$\mOp$ is defined on subsets of $W$ as
$\mOp{X} = \{ w \in W\ |\ \exists \alpha.\, \Covm{w}{\alpha} \subseteq
X\}$.
\begin{proposition}\label[prop]{prop:lupset-modal-operator}
  The operator~$\mOp$ is a monotone
  function~$\mOp : \LUp{(W)} \to \LUp{(W)}$ on localized up\hyp{}sets
\end{proposition}
\begin{proof}
  While monotonicity holds readily, we must show $\mOp{X}$ is a
  localized up\hyp{}set whenever $X$ is.

  To show $\mOp{X}$ is an up\hyp{}set, suppose $\Ri{w}{w'}$ and
  $w \in \mOp{X}$.
  This means for some~$\beta$,
  $w' \RiOp{}{} \Covm{w}{\beta} \subseteq X$.
  Due to the modal refinement condition, we know that for some
  $\beta'$, we have $\Covm{w'}{\beta'} \RinOp{}{} \beta$.
  Since~$X$ is an up\hyp{}set and $\beta'$ refines $\beta$, we also
  have $\beta' \subseteq X$, and thus $\Covm{w'}{\beta'} \subseteq X$,
  which is why $w' \in \mOp{X}$.
    
  To show $\mOp{X}$ is localizing, recollect that the modal
  localization effectively
  states~$\jOp{\mOp{X}} \subseteq \mOp{\jOp{X}}$.
  Since X is localizing, we know $\jOp{X} \subseteq X$, which implies
  $\jOp{\mOp{X}} \subseteq \mOp{X}$ since $\mOp$ is monotonic.
\end{proof}
A \emph{modal Heyting algebra}~$\tuple{H,m}$ is a Heyting algebra~$H$
accompanied by a monotone function~$m : U \to U$ on the carrier
set~$U$ underlying H.
An algebraic model~$\Mod[A] = \triple{H}{m}{V}$ of \CM consists
of a modal Heyting algebra~$\tuple{H,m}$ and a valuation
function~$V : \Atom \to U$.
As before, it can be shown that interpretation of
formulas~$\eval{-}^{\Mod[A]} : \Prop \to U$ and
contexts~$\eval{-}^{\Mod[A]} : \Ctx \to U$ in \CM are sound and
complete for the judgments~$\Gamma \vdash A$ in \CM by taking
$\eval{\MonPr{A}}^{\Mod[A]} = m(\eval{A}^{\Mod[A]})$ for the case of
modal formulas~$\MonPr{A}$ in \CM.

\begin{proposition}\label[prop]{prop:lupset-modal-algebra}
  Every modal cover model~$\Mod[M]=\triple{C}{\Covm{}{}}{V}$ of \CM
  determines an equivalent algebraic model
  $\Psh[\Mod[M]]=\triple{\Psh[C]}{\mOp}{V}$ of \CM, with $\Psh[C]$ as
  in \cref{prop:cm-to-am}, s.t.
  $\truth{A}^{\Mod[M]} = \eval{A}^{\Psh[\Mod[M]]}$ for all \CM
  formulas~$A$.
\end{proposition}
\begin{proof}
  Follows from \cref{prop:cm-to-am,prop:lupset-modal-operator}
  and the observation
  $\eval{\MonPr{A}}^{\Psh[\Mod[M]]} = \mOp{\eval{A}^{\Psh[\Mod[M]]}}$.
\end{proof}
%
%

%

\section{Modal Cover Semantics for IMLs}\label{sec:imls}

\subsection{Minimal Lax Logic}

The language of \SL extends that of \IPL with a unary
connective~$\LaxPr$ known as the lax modality.
A modal formula~$\LaxPr{A}$ may be intuitively understood as denoting
the truth of formula $A$ qualified by some constraint, i.e. ``possibly
$A$''.
The logic \SL admits the characteristic
axiom~$\SA : A \AndPr \LaxPr{B} \ImpPr \LaxPr{(A \AndPr B)}$, which
states that if $A$ is true and $B$ is possibly true, then
both $A$ and $B$ are possibly true.
The proof rules for \SL extend those of \IPL with a
rule~\SL/$\LaxPr$\nbhyp{}Map (for ``mapping'').
\begin{align*}
  \Prop\ \ A,B  :=  \ldots\ |\ \LaxPr{A} \qquad
  \inferrule[\SL/$\LaxPr$\nbhyp{}Map]{%
      \Gamma  \vdash \LaxPr{A} \\
      \ExtCtx{\Gamma}{A} \vdash B
    }{%
      \Gamma \vdash \LaxPr{B}
  }
\end{align*}

An \emph{\SL algebra} is a modal Heyting algebra~$\tuple{H,m}$ where
the monotone function~$m : U \to U$ satisfies the
inequality~$a \times m(b) \leq m(a \times b)$, for all~$a,b \in U$.
We may equivalently characterize an \SL algebra ``equationally'' (as
in \cite[Definition 4]{AlechinaMPR01}) by dropping the monotonicity
condition on the function~$m$ in favor of an additional
inequality~$m(a) \leq m(a + b)$.
An algebraic model~$\Mod[A] = \triple{H}{m}{V}$ of \SL consists of an
\SL algebra~$\tuple{H,m}$ and a valuation function~$V : \Atom \to U$
mapping atoms to the carrier set~$U$ underlying~$H$.
The interpretation of formulas in \SL can be given by extending the
interpretation~~$\eval{-} : \Prop \to U$ of formulas in \IPL with
$\eval{\LaxPr{A}}^{\Mod[A]} = m(\eval{A}^{\Mod[A]})$ for the case of
modal formulas~$\LaxPr{A}$ in \SL.

It can further be shown by induction that if a
judgment~$\Gamma \vdash A$ is derivable in \SL, then
$\eval{\Gamma}^{\Mod[A]} \leq \eval{A}^{\Mod[A]}$ holds for all
algebraic models~$\Mod[A]$ of \SL.
The interesting case is that of
Rule~\SL/$\LaxPr$\nbhyp{}Map.
By applying the induction hypothesis to the premises of the rule, we
obtain the inequalities~$\eval{\Gamma} \leq m{\eval{A}}$ (IH.1) and
$\eval{\Gamma} \times \eval{A} \leq \eval{B}$ (IH.2).
It follows from IH.1 that
$\eval{\Gamma} \leq (\eval{\Gamma} \times m{\eval{A}}) \leq
m(\eval{\Gamma} \times \eval{A})$ for an \SL algebra, which when
combined with IH.2 and monotonicity of $m$ gives us the
inequality~$\eval{\Gamma} \leq m{\eval{B}}$ as desired.

A modal cover model~$\Mod[M] = \triple{C}{\Covm{}{}}{V}$ of \SL
consists of a modal cover system~$\tuple{C,\Covm{}{}}$ and a valuation
function~$V : \Atom \to \LUp{(W)}$, where the modal covering
relation~$\Covm{}{}$ satisfies, in addition to the usual modal
refinement and localization conditions, a \emph{modal inclusion}
condition stated below:
\begin{itemize}
\item \emph{Modal Inclusion}: If $\Covm{w}{\alpha}$, then $\Rin{\{w\}}{\alpha}$
\end{itemize}
The truth of modal formulas for an arbitrary modal cover
model~$\Mod[M]$ of \SL is given as before for \CM by extending the
satisfaction relation to modal formulas~$\LaxPr{A}$ in a manner that
ensures~$\truth{\LaxPr{A}}^{\Mod[M]} = \mOp{\truth{A}^{\Mod[M]}}$.
\begin{equation*}
  \begin{array}{l@{\;\Vdash\;} @{\;} l @{\;\text{iff}\;}c@{\;} l}
    \Mod[M],w & \LaxPr{A} & & \exists \beta.\, \Covm{w}{\beta}\ \text{and}\ \forall v \in \beta.\, \Mod[M],v \Vdash A
  \end{array}
\end{equation*}

\begin{proposition}\label[prop]{prop:lupset-sl-algebra}
  Every modal cover model~$\Mod[M]=\triple{C}{\Covm{}{}}{V}$ of \SL
  determines an equivalent algebraic
  model~$\Psh[\Mod[M]]=\triple{\Psh[C]}{\$\mOp}{V}$ of \SL, whose
  underlying Heyting algebra~$\Psh[C]$ is given by localized
  up\hyp{}sets as in \cref{prop:cm-to-am}, such that
  $\truth{A}^{\Mod[M]} = \eval{A}^{\Psh[\Mod[M]]}$ and
  $\truth{\Gamma}^{\Mod[M]} = \eval{\Gamma}^{\Psh[\Mod[M]]}$ for all
  formulas~$A$ and contexts~$\Gamma$ in \SL.
\end{proposition}
\begin{proof}
  Every modal cover system~$\tuple{C,\Covm{}{}}$ determines a modal
  Heyting algebra~$\tuple{\Psh[C],\mOp}$, as in the proof of
  \cref{prop:lupset-modal-algebra} due to
  \cref{prop:cm-to-am,prop:lupset-modal-operator}.
  To show $\tuple{\Psh[C],\mOp}$ is also an \SL algebra, it remains to
  show $X \cap \mOp{Y} \subseteq \mOp{(X \cap Y)}$ for
  all~$X,Y \in \LUp{(W)}$, which we achieve using the modal inclusion
  condition.
  
  Observe that the interpretation of formulas in $\Psh[\Mod[M]]$
  readily satisfies the equality
  $\eval{\LaxPr{A}}^{\Psh[\Mod[M]]} = \mOp{\eval{A}^{\Psh[\Mod[M]]}}$ by
  definition.
  As a result, we can once again show that
  $\eval{A}^{\Psh[\Mod[M]]} = \truth{A}^{\Mod[M]}$ and
  $\eval{\Gamma}^{\Psh[\Mod[M]]} = \truth{\Gamma}^{\Mod[M]}$ by
  induction on $A$ and $\Gamma$.
  For the case of modal formulas~$\LaxPr{A}$, we observe that
  $\eval{\LaxPr{A}}^{\Psh[\Mod[M]]} = \mOp{\eval{A}^{\Psh[\Mod[M]]}} =
  \mOp{\truth{A}^{\Mod[M]}} = \truth{\LaxPr{A}}^{\Mod[M]}$ .
\end{proof}
\begin{proposition}[Soundness for \SL]\label[proposition]{prop:sl-sound}
If \,$\Gamma \vdash A$, then $\Gamma \satisfies A$.
\end{proposition}
\begin{proof}
  By soundness of \SL for its algebraic models and
  \cref{prop:lupset-sl-algebra}, we have $\Gamma \vdash A$ implies
  $\eval{\Gamma}^{\Mod[M]} \subseteq \eval{A}^{\Mod[M]}$ for all
  modal cover models~$\Mod[M]$.
  Since the algebraic interpretation of formulas and contexts in \SL
  is equivalent to their respective truth sets, we also have
  $\truth{\Gamma}^{\Mod[M]} \subseteq \truth{A}^{\Mod[M]}$, and thus
  $\Gamma \satisfies A$.
\end{proof}

As before with \CM, to prove completeness for \SL we construct a
canonical modal cover model~$\Mod[N]$ that equates entailment of
formulas in the model~$\Mod[N]$ to provability in \SL.
For this purpose, let us define a modal cover
system~$C_{\text{\SL}} = \tuple{C_{\text{\IPL}},\CovmSL{}{}}$ coupling
the cover system~$C_{\text{\IPL}}$ (reproducing \cref{def:cov-sys-ipl}
in the language of \SL) with the modal covering
relation~$\CovmSL{}{}\ \subseteq \Ctx \times \Pow{(\Ctx)}$ defined
inductively below.
\begin{mathpar}
  \inferrule[]{%
    \Gamma \in \AntTh{(\LaxPr{A})}
  }{%
    \CovmSL{\Gamma}{\{\ExtCtx{\Gamma}{A}\}}
  }

  \inferrule[]{%
    \Gamma \in \AntTh{(\BotPr)}}{%
    \CovmSL{\Gamma}{\emptyset}
  }

  \inferrule[]{%
    \Gamma \in \AntTh{(A \OrPr B)}\\
    \CovmSL{\ExtCtx{\Gamma}{A}}{\alpha_1} \\
    \CovmSL{\ExtCtx{\Gamma}{B}}{\alpha_2} \\
  }{%
    \CovmSL{\Gamma}{\alpha_1 \union \alpha_2}
  }
\end{mathpar}
As before with the relation~$\CovmCM{}{}$, the relation~$\CovmSL{}{}$
can be shown to satisfy the modal refinement and localization
conditions.
In contrast to $\CovmCM{}{}$, however, $\CovmSL{}{}$ crucially also
satisfies the modal inclusion condition.
This is because all contexts~$\Gamma'$ in a
cover~$\CovmSLOp{\alpha}{\Gamma}$ subsume~$\Gamma$, i.e.
$\Gamma \subseteq \Gamma'$, and thus $\Rin{\{\Gamma\}}{\alpha}$.

\begin{lemma}[Truth Lemma]\label[lem]{lem:sl-truth-lemma}
  The tuple~$\Mod[N] = \tuple{C_{\text{\SL}},\AntTh}$ is a cover model
  of \SL s.t. $\truth{A}^{\Mod[N]} = \AntTh{(A)}$.
\end{lemma}
\begin{theorem}[Completeness for \SL]\label[thm]{thm:sl-complete}
If \,$\Gamma \satisfies A$, then $\Gamma \vdash A$.
\end{theorem}
\begin{proof}
  By repeating the argument in \cref{thm:cm-complete}, now using
  \cref{lem:sl-truth-lemma}.
\end{proof}

\subsection{Propositional Lax Logic}

The logic \LL extends the axioms of \SL with the axioms~$\RA : A
\ImpPr \LaxPr{A}$ and~$\JA : \LaxPr{\LaxPr{A}} \ImpPr \LaxPr{A}$.
The language of \LL extends that of \IPL with a unary
connective~$\LaxPr$, as with the language of \SL, while the proof
system for \LL extends that of \IPL with the
rules~\LL/$\LaxPr$\nbhyp{}Intro and \LL/$\LaxPr$\nbhyp{}Bind (for
``binding'') given below.
\begin{align*}
  \Prop\ \ A,B  :=  \ldots\ |\ \LaxPr{A} \qquad
  \inferrule[\LL/$\LaxPr$\nbhyp{}Intro]{%
      \Gamma  \vdash A
    }{%
      \Gamma \vdash \LaxPr{A}
    } \qquad
    \inferrule[\LL/$\LaxPr$\nbhyp{}Bind]{%
      \Gamma  \vdash \LaxPr{A} \\
      \ExtCtx{\Gamma}{A} \vdash \LaxPr{B}
    }{%
      \Gamma \vdash \LaxPr{B}
    }
\end{align*}

A \emph{\LL algebra} is an \SL algebra~$\tuple{H,m}$ where the
monotone function~$m : U \to U$ is inflationary and idempotent,
i.e. it additionally satisfies the inequalities~$a \leq m(a)$ and
$m(m(a)) \leq m(a)$, for all $a \in U$.
We may equivalently characterize a \LL algebra without reference to
\SL algebras as a Heyting algebra~$H$ accompanied by a nucleus
operator~$m$.
An algebraic model~$\Mod[A] = \triple{H}{m}{V}$ of \LL consists of a
\LL algebra~$\tuple{H,m}$ and a valuation function~$V : \Atom \to U$,
where the interpretation of modal formulas is given by
$\eval{\LaxPr{A}}^{\Mod[A]} = m(\eval{A}^{\Mod[M]})$.
It is known moreover that $\Gamma \vdash A$ in \LL if and only if
$\eval{\Gamma}^{\Mod[A]} \leq \eval{A}^{\Mod[A]}$ for all algebraic
models~$\Mod[M]$ of \LL~\cite{AlechinaMPR01,Goldblatt11a}.

A modal cover model~$\Mod[M] = \triple{C}{\Covm{}{}}{V}$ of \LL is
a modal cover model of \SL that additionally satisfies the modal
\emph{identity} and \emph{transitivity} conditions stated below:
\begin{itemize}
\item \emph{Modal Identity}: $\Covm{w}{\{w\}}$
\item \emph{Modal Transitivity}: If $\Covm{w}{\alpha}$ and for all $v \in
  \alpha$ there exists an $\alpha_v$ such that $\Covm{v}{\alpha_v}$,
  then $\Covm{w}{\Union_{v \in \alpha} \alpha_v}$
\end{itemize}
The truth of modal formulas for a model~$\Mod[M]$ of \LL is given
identically to \SL ensuring
$\truth{\LaxPr{A}}^{\Mod[M]} = \mOp{\truth{A}^{\Mod[M]}}$.

As before with \SL, every modal cover
model~$\Mod[M]=\triple{C}{\Covm{}{}}{V}$ of \LL determines an
equivalent algebraic model~$\Psh[\Mod[M]]=\triple{\Psh[C]}{\$\mOp}{V}$
of \LL, i.e. $\truth{A}^{\Mod[M]} = \eval{A}^{\Psh[\Mod[M]]}$ and
$\truth{\Gamma}^{\Mod[M]} = \eval{\Gamma}^{\Psh[\Mod[M]]}$ for all
formulas~$A$ and contexts~$\Gamma$ in \LL.
The operator~$\mOp$ is a nucleus on $\LUp{(W)}$---the carrier set of
the Heyting algebra~$\Psh[C]$---since its underlying modal covering
relation~$\Covm{}{}$ satisfies the same conditions (refinement,
inclusion, identity and transitivity) imposed on the local covering
relation~$\Covp{}{}$ underlying the nucleus operator~$\jOp$ (c.f.
\cref{prop:jop-properties}).

\begin{proposition}[Soundness for \LL]\label[proposition]{prop:ll-sound}
  If \,$\Gamma \vdash A$, then $\Gamma \satisfies A$.
\end{proposition}
\begin{proof}
  Every modal cover model of \LL determines an equivalent algebraic
  model of \LL, which is sound for \LL (repeating the argument in
  \cref{prop:sl-sound}).
  Thus modal cover semantics is as well sound.
\end{proof}

To prove completeness for \LL, we define a modal cover
system~$C_{\text{\LL}} = \tuple{C_{\text{\IPL}},\CovmLL{}{}}$ coupling
the cover system~$C_{\text{\IPL}}$ from earlier with the modal
covering relation~$\CovmLL{}{}\ \subseteq \Ctx \times \Pow{(\Ctx)}$
defined below.
\begin{mathpar}
  \inferrule[]{%
  }{%
    \CovmLL{\Gamma}{\{\Gamma\}}
  }
  
  \inferrule[]{%
    \Gamma \in \AntTh{(\LaxPr{A})}\\
    \CovmLL{\ExtCtx{\Gamma}{A}}{\alpha}
  }{%
    \CovmLL{\Gamma}{\alpha}
  }\\

  \inferrule[]{%
    \Gamma \in \AntTh{(\BotPr)}}{%
    \CovmLL{\Gamma}{\emptyset}
  }

  \inferrule[]{%
    \Gamma \in \AntTh{(A \OrPr B)}\\
    \CovmLL{\ExtCtx{\Gamma}{A}}{\alpha_1} \\
    \CovmLL{\ExtCtx{\Gamma}{B}}{\alpha_2} \\
  }{%
    \CovmLL{\Gamma}{\alpha_1 \union \alpha_2}
  }
\end{mathpar}
The relation~$\CovmLL{}{}$ readily satisfies the modal identity
condition by definition, while modal transitivity and the remaining
conditions can be shown by induction on its definition.
The definition of $\CovmLL{}{}$ ensures that
$\jOp{X} \subseteq \mOp{X}$ for all subsets~$X \subseteq \Ctx$.
The modal localization condition follows as a result, since
$\jOp{\mOp{X}} \subseteq \mOp{\mOp{X}} = \mOp{X} \subseteq
\mOp{\jOp{X}}$, given $\mOp{}$ and $\jOp$ are both closure operators
on $\Pow{(\Ctx)}$.

\begin{theorem}[Completeness for \LL]\label[thm]{thm:ll-complete}
If \,$\Gamma \satisfies A$, then $\Gamma \vdash A$.
\end{theorem}
\begin{proof}
  By showing the truth lemma for
  $\Mod[N] = \tuple{C_{\text{\LL}},\AntTh}$ and repeating the argument
  in \cref{thm:sl-complete}.
\end{proof}

\subsection{Dual-Context \CKBox}

The language of \CKBox extends that of \IPL with a unary
connective~$\BoxPr$ known as the box modality.
A modal formula~$\BoxPr{A}$ may be read as ``necessarily A'' and
intuitively understood as asserting that $A$ is valid,
i.e. universally true.
The logic \CKBox admits the \emph{necessitation} rule, which states
that if $A$ is valid then so is $\Box{A}$, and the characteristic
axiom~$\KA : \BoxPr{(A \ImpPr B)} \ImpPr \BoxPr{A} \ImpPr \BoxPr{B}$.
A \emph{dual\hyp{}context} sequent\hyp{}style proof system for \CKBox,
denoted \DCKBox, is given using
judgments~$\DualCtx{\Delta}{\Gamma} \vdash A$ indexed by two
contexts~$\Delta$ and $\Gamma$.
The ``global'' context~$\Delta$ consists of formulas that are assumed
to be valid, while the usual ``local'' context~$\Gamma$ consists of
formulas that are assumed to be true for some specific world.
The proof rules for the non\nbhyp{}modal fragment can be given as
before for \IPL by leaving the global context untouched (see
\cref{sec:app}).
The proof rules for the modal fragment are given by the
rules~\DCKBox/$\BoxPr$\nbhyp{}Intro
and~\DCKBox/$\BoxPr$\nbhyp{}Elim defined below.
\begin{align*}
  \Prop\ \ A,B  :=  \ldots\ |\ \BoxPr{A} \qquad
  \inferrule[\DCKBox/$\BoxPr$\nbhyp{}Intro]{%
      \DualCtx{\EmptyCtx}{\Delta} \vdash A
    }{%
      \DualCtx{\Delta}{\Gamma} \vdash \BoxPr{A}
    } \qquad
    \inferrule[\DCKBox/$\BoxPr$\nbhyp{}Elim]{%
      \DualCtx{\Delta}{\Gamma} \vdash \BoxPr{A} \\
       \DualCtx{\ExtCtx{\Delta}{A}}{\Gamma} \vdash B
    }{%
      \DualCtx{\Delta}{\Gamma} \vdash B
    }
\end{align*}

A \emph{\CKBox algebra} is a modal Heyting algebra~$\tuple{H,m}$ where
the monotone function~$m : U \to U$ preserves all finite meets, i.e.
it satisfies the equations~$m (1) = 1$ and~$m(a \times b) = m(a)
\times m(b)$, for all~$a,b \in U$.
An algebraic model~$\Mod[A] = \triple{H}{m}{V}$ of \DCKBox consists of
a \CKBox algebra~$\tuple{H,m}$ and a valuation
function~$V : \Atom \to U$.
The interpretation~$\eval{-}^{\Mod[A]}$ of formulas and contexts in
the set~$U$ is defined as before for the previous logics, where the
interpretation of modal formulas is given as
$\eval{\BoxPr{A}}^{\Mod[A]} = m(\eval{A}^{\Mod[A]})$.

For some algebra~$\Mod[A]$ of \CKBox, a formula~$A$ is said to be
\emph{algebraically valid} iff $\eval{A} = 1$ and \emph{algebraically
true} for an element (as opposed to world) $u \in U$ iff $u \leq
\eval{A}$.
Consequentially, a formula~$A$ is algebraically valid iff it is
algebraically true for all elements.
The requirement that $m$ preserves all finite meets allows us to show
that the necessitation rule and axiom~$\KA$ are algebraically sound
principles.
The equation~$m(1) = 1$ allows us to show that the necessitation rule
is algebraically sound: if a formula~$A$ is algebraically valid,
meaning $\eval{A} = 1$, then $\eval{\BoxPr{A}} = m(\eval{A}) = m(1) =
1$, and thus $\BoxPr{A}$ is also algebraically valid.
Similarly, the equation~$m(a) \times m(b) = m(a \times b)$ allows us
to show that axiom~$\KA : \BoxPr{(A \ImpPr B)} \ImpPr \BoxPr{A} \ImpPr
\BoxPr{B}$ is algebraically valid, since it implies that the
inequality~$m(a \ImpPr b) \times m(a) \leq m(b)$ holds for all
elements~$a,b \in U$.

If a dual\hyp{}context judgment~$\DualCtx{\Delta}{\Gamma} \vdash A$
is derivable in \DCKBox, then the
inequality~$m(\eval{\Delta}^{\Mod[A]}) \times \eval{\Gamma}^{\Mod[A]}
\leq \eval{A}^{\Mod[A]}$ must hold for all algebraic models~$\Mod[A]$
of \DCKBox.
This can be observed readily by induction on the derivation of
judgment by using the equality~$1 = m (1)$ for the case of
Rule~\DCKBox/$\BoxPr$\nbhyp{}Intro and $m(a) \times m(b) = m(a \times
b)$ for the case of Rule~\DCKBox/$\BoxPr$\nbhyp{}Elim.
Alternatively, we may also appeal to the soundness of the categorical
interpretation of \DCKBox \cite[Section 6.2]{Kavvos17}, where the
modality~$\BoxPr{}$ is interpreted as an endofunctor (analogous to
$m$) preserving finite products ($\times$) on a cartesian\hyp{}closed
category ($H$).

A modal cover model~$\Mod[M] = \triple{C}{\Covm{}{}}{V}$ of \DCKBox
consists of a modal cover system~$\tuple{C,\Covm{}{}}$ and a valuation
function~$V$, where the modal covering relation~$\Covm{}{}$ satisfies,
in addition to the usual modal refinement and localization
conditions, the modal \emph{seriality} and \emph{confluence}
conditions stated below:
\begin{itemize}
\item \emph{Modal Seriality}: For all $w \in W$, there exists an
  $\alpha$ such that $\Covm{w}{\alpha}$
\item \emph{Modal Confluence}: If $\Covm{w}{\alpha}$ and
  $\Covm{w}{\beta}$, then there exists a $\gamma$
  s.t. $\Covm{w}{\gamma}$ and $\Rin{\alpha}{\gamma} \RinOp{}{\beta}$
\end{itemize}
The truth of \CKBox formulas is defined as before
ensuring~$\truth{\BoxPr{A}}^{\Mod[M]} = \mOp{\truth{A}^{\Mod[M]}}$, by
extending the definition of the satisfaction relation for \IPL with
the following case for modal formulas:
\begin{equation*}
  \begin{array}{l@{\;\Vdash\;} @{\;} l @{\;\text{iff}\;}c@{\;} l}
    \Mod[M],w & \BoxPr{A} & & \exists \beta.\, \Covm{w}{\beta}\ \text{and}\ \forall v \in \beta.\, \Mod[M],v \Vdash A
  \end{array}
\end{equation*}
Entailment in a modal cover model~$\Mod[M]$ of \DCKBox is defined
by incorporating dual contexts as $\DualCtx{\Delta}{\Gamma}
\satisfies_{\Mod[M]} A$ if and only if $\mOp{\truth{\Delta}^{\Mod[M]}}
\cap \truth{\Gamma}^{\Mod[M]} \subseteq \truth{A}^{\Mod[M]}$.
The application of the modal operator~$\mOp$ to the interpretation of
the global context~$\Delta$ ensures that all the assumptions in
$\Delta$ are all valid.
Continuing a previous convention, we will write
$\DualCtx{\Delta}{\Gamma} \satisfies A$ to mean
$\DualCtx{\Delta}{\Gamma} \satisfies_{\Mod[M]} A$ for all
models~$\Mod[M]$.

A cover model~$\Mod[M]=\triple{C}{\Covm{}{}}{V}$ of \DCKBox
determines an equivalent algebraic
model~$\Psh[\Mod[M]]=\triple{\Psh[C]}{\$\mOp}{V}$ of \DCKBox such that
$\truth{A}^{\Mod[M]} = \eval{A}^{\Psh[\Mod[M]]}$ and
$\truth{\Gamma}^{\Mod[M]} = \eval{\Gamma}^{\Psh[\Mod[M]]}$ for all
formulas~$A$ and contexts~$\Gamma$ in \CKBox.
Recollect that the maximal element of the Heyting algbra~$\Psh[C]$ is
$W$ and its meets are given by the intersection~$\cap$ of
localized\hyp{}up sets.
The modal seriality and confluence condition respectively allow us to
show that the operator~$\mOp$ satisfies the equations~$W = \mOp{W}$
and $\mOp{X} \cap \mOp{Y} = \mOp{(X \cap Y)}$ desired of a \CKBox
algebra.
The inequalities~$\mOp{W} \subseteq W$ and $\mOp{(X \cap Y)} \subseteq
\mOp{X} \cap \mOp{Y}$ hold readily since $W$ is maximal and $\mOp$ is
monotonic.
The converse~$W \subseteq \mOp{W}$ of the former follows from the
seriality condition: any $w \in W$ has some modal
cover~$\alpha \subseteq W$ and thus we also have $w \in \mOp{W}$.
Similarly for any~$X,Y \in \LUp(W)$, the inequality~$\mOp{X} \cap
\mOp{Y} \subseteq \mOp{(X \cap Y)}$ follows from the confluence
condition: if $w \in \mOp{X} \cap \mOp{Y}$, then for some $\alpha,
\beta$ we have $\Covm{w}{\alpha} \subseteq X$ and $\Covm{w}{\beta}
\subseteq Y$, from which we obtain a cover $\gamma$ of $w$
s.t. $\Rin{\alpha}{\gamma} \RinOp{}{\beta}$ by applying confluence.
Since $X$ and $Y$ are up\hyp{}sets, it must be the case that $\gamma
\subseteq X$ and $\gamma \subseteq Y$, which means $\Covm{w}{\gamma}
\subseteq (X \cap Y)$ and thus $w \in \mOp{(X \cap Y)}$.
Altogether we have shown the desired equations.

To prove completeness for \DCKBox, we define a modal cover system by taking
pairs of contexts, i.e. the set~$\Ctx \times \Ctx$, for worlds,
point\hyp{}wise context inclusion for the preorder relation, and the
below inductively defined relation~$\CovmDCKBox{}{}$ for the modal
covering relation.
The local covering relation is given by re\hyp{}defining the
relation~$\CovpIPL{}{}$ for dual\hyp{}contexts by prepending a global
context~$\Delta$ uniformly to all the cases.
\begin{mathpar}
  \inferrule[]{%
  }{%
    \CovmDCKBox{\DualCtx{\Delta}{\Gamma}}{\{\DualCtx{\EmptyCtx}{\Delta}\}}
  }
  
  \inferrule[]{%
    \DualCtx{\Delta}{\Gamma} \in \AntTh{(\BoxPr{A})}\\
    \CovmDCKBox{\DualCtx{\ExtCtx{\Delta}{A}}{\Gamma}}{\alpha}
  }{%
    \CovmDCKBox{\Gamma}{\alpha}
  }\\

  \inferrule[]{%
    \DualCtx{\Delta}{\Gamma} \in \AntTh{(\BotPr)}}{%
    \CovmDCKBox{\DualCtx{\Delta}{\Gamma}}{\emptyset}
  }

  \inferrule[]{%
    \DualCtx{\Delta}{\Gamma} \in \AntTh{(A \OrPr B)}\\
    \CovmDCKBox{\DualCtx{\Delta}{\ExtCtx{\Gamma}{A}}}{\alpha_1} \\
    \CovmDCKBox{\DualCtx{\Delta}{\ExtCtx{\Gamma}{B}}}{\alpha_2} \\
  }{%
    \CovmDCKBox{\DualCtx{\Delta}{\Gamma}}{\alpha_1 \union \alpha_2}
  }
\end{mathpar}

\begin{theorem}[Soundness and Completeness for \CKBox]\label[thm]{thm:ckbox-sound-complete}
$\DualCtx{\Delta}{\Gamma} \vdash A$ if and only if $\DualCtx{\Delta}{\Gamma} \satisfies A$.
\end{theorem}
\begin{proof}
Using the above arguments, by re\hyp{}establishing the truth lemma
once again for \DCKBox.
\end{proof}

\subsection{Dual-Context \CSFourBox}

The logic~\CSFourBox extends \CKBox with
axioms~$\TA : \BoxPr{A} \ImpPr A$
and~$\FourA : \BoxPr{A} \ImpPr \BoxPr{\BoxPr{A}}$, and its language is
identical to that of \CKBox.
A dual\hyp{}context sequent\hyp{}style proof system for \CSFourBox,
denoted \DCSFourBox, is given by extending the proof system~\DCKBox
above with an additional hypothesis rule~\DCSFourBox/Global\nbhyp{}Hyp
and by modifying the introduction rule for $\BoxPr$.
Observe how the introduction rules~\DCSFourBox/$\BoxPr$\nbhyp{}Intro
and~\DCKBox/$\BoxPr$\nbhyp{}Intro have different premises, and that
the elimination rules~\DCSFourBox/$\BoxPr$\nbhyp{}Elim and
\DCKBox/$\BoxPr$\nbhyp{}Elim are identical.
\begin{mathpar}
  \Prop\ \ A,B  :=  \ldots\ |\ \BoxPr{A} \qquad

  \inferrule[\DCSFourBox/Global\nbhyp{}Hyp]{%
      A \in \Delta
    }{%
      \DualCtx{\Delta}{\Gamma} \vdash A
    }\qquad
  \inferrule[\DCSFourBox/$\BoxPr$\nbhyp{}Intro]{%
      \DualCtx{\Delta}{\EmptyCtx} \vdash A
    }{%
      \DualCtx{\Delta}{\Gamma} \vdash \BoxPr{A}
    } \qquad
    \inferrule[\DCSFourBox/$\BoxPr$\nbhyp{}Elim]{%
      \DualCtx{\Delta}{\Gamma} \vdash \BoxPr{A} \\
       \DualCtx{\ExtCtx{\Delta}{A}}{\Gamma} \vdash B
    }{%
      \DualCtx{\Delta}{\Gamma} \vdash B
    }

\end{mathpar}

A \emph{\CSFourBox algebra} is a \CKBox algebra~$\tuple{H,m}$ whose
meet\hyp{}preserving monotone function~$m : U \to U$ is an
\emph{interior} operator on the carrier set~$U$ underlying the Heyting
algebra~$H$.
In other words, the function~$m$ is deflationary and idempotent, which
means it satisfies the inequalities~$m(a) \leq a$
and~$m(a) \leq m(m(a))$ for all elements~$a \in U$.
Akin to the interpretation of \CKBox, we interpret modal formulas
as~$\eval{\BoxPr{A}}^{\Mod[A]} = m(\eval{A}^{\Mod[A]})$, given an
algebraic model~$\Mod[A]$ of \CSFourBox.
The algebraic validity of axioms~$\TA$ and~$\FourA$ readily follow
respectively from the requirement that $m$ is deflationary and
idempotent.
It can also be shown that if $\DualCtx{\Delta}{\Gamma} \vdash A$ is
derivable in \DCSFourBox, then
$m(\eval{\Delta}^{\Mod[A]}) \times \eval{\Gamma}^{\Mod[A]} \leq
\eval{A}^{\Mod[A]}$ must hold for all algebraic models~$\Mod[A]$ of
\CSFourBox.

A modal cover model~$\Mod[M] = \triple{C}{\Covm{}{}}{V}$ of
\DCSFourBox is a modal cover model of \DCKBox whose modal covering
relation~$\Covm{}{}$ additionally satisfies the modal
\emph{extraction} and \emph{duplication} conditions stated below:
\begin{itemize}
\item \emph{Modal Extraction}: If $\Covm{w}{\alpha}$, then there
  exists a $\beta$ such that
  $\Covp{w}{\beta} \RinOp{}{\alpha}$
\item \emph{Modal Duplication}: If $\Covm{w}{\alpha}$, then
  for all $v \in
  \alpha$ there exists an $\alpha_v$ such that $\Covm{v}{\alpha_v} \RinOp{}{\{v\}}$
\end{itemize}
The truth of modal formulas for a model~$\Mod[M]$ of \DCSFourBox is
given identically to \DCKBox, once again
ensuring~$\truth{\BoxPr{A}}^{\Mod[M]} = \mOp{\truth{A}^{\Mod[M]}}$.
The entailment
relation~$\DualCtx{\Delta}{\Gamma} \satisfies_{\Mod[M]} A$ is as well
identical to \DCKBox.

The modal extraction and duplication condition respectively allow us
to show that the operator~$\mOp$ satisfies the inequalities~$\mOp{X}
\subseteq X$ and~$\mOp{X} \subseteq \mOp{\mOp{X}}$ for any localized
up\hyp{}set $X$.
To show both inequalities, let us suppose an arbitrary~$w \in
\mOp{X}$, which means for some $\alpha$ we have $\Covm{w}{\alpha}
\subseteq X$.
For the former inequality, to show that $\mOp$ is deflationary,
observe that the modal extraction condition gives us a subset~$\beta$
such that $\Covp{w}{\beta} \RinOp{}{\alpha}$.
Since $\RinOp{\beta}{\alpha} \subseteq X$, we have $\beta \subseteq X$
given X is an up\hyp{}set.
We thus have $\Covp{w}{\beta} \subseteq X$ and, since $X$ is
localizing, also that $w \in X$.
For the latter inequality, the modal duplication condition gives us
for every~$v \in \alpha$ a subset~$\alpha_v$ such that
$\Covm{v}{\alpha_v} \RinOp{}{\{v\}}$.
Since $\alpha \subseteq X$ and $\RinOp{\alpha_v}{\{v\}}$ for every~$v
\in \alpha$, we also have $\alpha_v \subseteq X$ given $X$ is an
up\hyp{}set.
This means~$\Covm{v}{\alpha_v} \subseteq X$, or equivalently $v \in
\mOp{X}$.
Altogether we have $\Covm{w}{\alpha} \subseteq \mOp{X}$, or
equivalently $w \in \mOp{\mOp{X}}$, which is the desired result.

As before, we can show that a cover model~$\Mod[M]$ of \DCSFourBox
determines an equivalent algebraic
model~$\Psh[\Mod[M]]=\triple{\Psh[C]}{\$\mOp}{V}$ of \DCSFourBox such
that $\truth{A}^{\Mod[M]} = \eval{A}^{\Psh[\Mod[M]]}$ and
$\truth{\Gamma}^{\Mod[M]} = \eval{\Gamma}^{\Psh[\Mod[M]]}$ for all
formulas~$A$ and contexts~$\Gamma$ in \CSFourBox.
This in turn gives us soundness of cover semantics for \DCSFourBox,
akin to \cref{prop:sl-sound,prop:ll-sound}.

Let us now
turn our attention to proving completeness of modal cover semantics
for \DCSFourBox.
To construct a canonical model, we define a modal covering
relation~$\CovmDCSFourBox{}{}$ by reproducing the definition of the
relation~$\CovmDCKBox{}{}$ and replacing its base case with the
following rule on the left.
Furthermore, we define a local covering relation~$\CovpDCSFourBox{}{}$
by reproducing the definition of the relation~$\CovpDCKBox{}{}$ (which
is itself an adaptation of $\CovpIPL{}{}$ to dual\hyp{}contexts) and
extending the definition with the following rule on the right.
\begin{mathpar}
  \inferrule[]{%
  }{%
    \CovmDCSFourBox{\DualCtx{\Delta}{\Gamma}}{\{\DualCtx{\Delta}{\EmptyCtx}\}}
  }

  \inferrule[]{%
    \DualCtx{\Delta}{\Gamma} \in \AntTh{(\BoxPr{A})}\\
    \CovpDCSFourBox{\DualCtx{\ExtCtx{\Delta}{A}}{\Gamma}}{\alpha}
  }{%
    \CovpDCSFourBox{\Gamma}{\alpha}
  }
\end{mathpar}
The rule on the left ensures that the truth set~$\truth{\BoxPr{A}}$
determined by the modal covering relation~$\CovmDCSFourBox{}{}$ is
equivalent to the set~$\AntTh{(\BoxPr{A})}$ of antecedents for
\DCSFourBox, and allows us to satisfy the modal duplication condition.
The rule on the right allows us to show that the modal extraction
condition holds.

\begin{theorem}[Soundness and Completeness for \CSFourBox]\label[thm]{thm:csfourbox-sound-complete}
$\DualCtx{\Delta}{\Gamma} \vdash A$ if and only if $\DualCtx{\Delta}{\Gamma} \satisfies A$.
\end{theorem}
\begin{proof}
Using the above arguments, by re\hyp{}establishing the truth lemma
once again for \DCSFourBox.
\end{proof}

\section{Other Related and Further Work}
\label{sec:discussion}

We have presented modal cover semantics as a conservative extension of
Goldblatt's relational cover semantics for IMLs and shown as examples
five IMLs which can be modeled using modal cover semantics.
We have shown that modal cover semantics retains the simplicity of
model construction in Kripke\hyp{}style semantics, while overcoming
its reliance on classical reasoning to prove completeness.

\textbf{Formalization in type theory}.
The results in this article have been formalized in the proof
assistant and dependently\hyp{}typed programming language
Agda~\cite{Agda2}, whose underlying core type theory is constructive.
Formalizing our results in Agda ensures that our results are indeed
constructive and do not accidentally rely upon classical reasoning
principles.
To encode cover models in type theory, we use a type~$X : \Type$ in
place of a set~$X$ and values~$x: X$ in place of elements~$x \in X$.
We encode subsets~$X,Y \subseteq W$ as functions~$X,Y : W \to \Type$,
and the inclusion~$X \subseteq Y$ as a function $\forall w.\, X w \to
Y w$.
The covering relation~${\Covp{}{}} \subseteq W \times \Pow{(W)}$ is
decomposed into a neighborhood ``directory''~$\N : W \to \Type$ and a
membership relation~${-\varepsilon_w-} : \, W \to \N{w}
\to \Type$.
A cover~$\Covp{w}{\alpha}$ is encoded by an element~$\alpha : \N{w}$,
where a world~$v : W$ in our encoding satisfies the
relation~$v\, {\varepsilon_w}\, \alpha$ if and only if there exists a
world~$v \in W$ such that $v \in \alpha$.
We refer the reader to the accompanying formalization in Agda for
examples and further details.

\textbf{Normalization}.
The completeness proofs in the previous sections can be readily
refined to give normalization algorithms for proofs in the natural
deduction systems of the respective logics.
The definition of normal and neutral forms for this purpose can be
found in \cref{sec:app}, where each inference rule has been carefully
defined to satisfy the subformula property.
The normalization algorithms are implemented using the technique of
Normalization by Evaluation, and can be found in the accompanying Agda
formalization.

\begin{theorem}[Normalization]
  Every judgment derivable in the proof systems for~\CM, \SL, \LL,
  \CKBox, and \CSFourBox, has a derivation in normal form that can be
  constructed when given the original derivation.
\end{theorem}
\begin{proof}
  By refining the statement of the truth lemma, for example in
  \cref{lem:cm-truth-lemma} for \CM, as follows: the tuple~$\Mod[N] =
  \tuple{C_{\text{\CM}},\NfTh}$ is a cover model of \CM such that for
  every formula~$A$, we have $\truth{A}^{\Mod[N]} \subseteq
  \NfTh{(A)}$ and $\NeTh{(A)} \subseteq \truth{A}^{\Mod[N]}$, where
  $\NfTh{(A)} = \{\Gamma \in \Ctx\ |\ \Gamma \vdashNf A\}$ and
  $\NeTh{(A)} = \{\Gamma \in \Ctx\ |\ \Gamma \vdashNe A\}$.
\end{proof}

\textbf{Other models of intuitionistic modal logics}.
A body of work that is closely related to our approach is
\emph{neighborhood semantics} for intuitionistic modal
logics~\cite{ArecesF09,DalmonteGO20,Dalmonte22,Groot26}.
The modalities~$\BoxPr{}$ and $\DiaPr{}$ are modeled using a
neighborhood function~$\mcN : W \to \Pow{(\Pow{(W)})}$, for example in
\cite[Definition 4.1]{Dalmonte22}, as follows:
\begin{equation*}
  \begin{array}{l@{\;\Vdash\;} @{\;} l @{\;\text{iff}\;}c@{\;} l}
    \Mod[M], w & \BoxPr{A} & & \forall w'.\, \Ri{w}{w'}\ \text{implies}\ \exists \alpha.\, \alpha \in \mcN{(w')}\ \text{and}\  \forall v.\, v \in \alpha\ \text{implies}\ \Mod[M],v \Vdash A \\
    \Mod[M], w & \DiaPr{A} & & \forall w'.\, \Ri{w}{w'}\ \text{implies}\ \forall \alpha.\, \alpha \in \mcN{(w')}\ \text{implies}\  \exists v.\, v \in \alpha\ \text{and}\ \Mod[M],v \Vdash A
  \end{array}
\end{equation*}
While these clauses can be presented equivalently using a modal
covering relation~$\Covm{}{}\ \subseteq W \times \Pow{(W)}$, a key
difference is that we have used a clause resembling the former to
model \emph{all} modalities, including $\BoxPr{}$ and $\DiaPr{}$,
alike.
Moreover, another difference is the treatment of the positive
connectives in these works.
They do not use a local covering relation~$\Covp{}{}$ and instead
follow the usual Kripke\hyp{}style approach as follows:
\vspace{-\abovedisplayskip}
\begin{center}
\begin{minipage}{.4\linewidth}
\centering
\begin{equation*}
  \begin{array}{l@{\;\Vdash\;} @{\;} l @{\;\text{iff}\;}c@{\;} l}
    \Mod[M],w & \BotPr & & \text{false}
  \end{array}
\end{equation*}
\end{minipage}%
\begin{minipage}{.6\linewidth}
\centering
\begin{equation*}
  \begin{array}{l@{\;\Vdash\;} @{\;} l @{\;\text{iff}\;}c@{\;} l}
    \Mod[M],w & A \OrPr B & & \Mod[M],w \Vdash A\ \text{or}\ \Mod[M],w \Vdash B
  \end{array}
\end{equation*}
\end{minipage}
\end{center}
The completeness proofs (see \cite[Lemma 4.4]{Dalmonte22}) rely on
prime sets as a result, and are thus not constructive.
Niki~\cite{Niki22} proposes an extension of Beth semantics to account
for modal connectives using an accessibility relation, which we may
perhaps call Beth\hyp{}Kripke\hyp{}style relational semantics.
Niki's semantics generalizes Kripke\hyp{}style relational semantics in
its treatment of the positive connectives.
However, the corresponding completeness proofs (see for
e.g.~\cite[``Proof of Theorem 2.9'' in the Appendix]{Niki22}) are not
constructive since they rely on embedding Kripke\hyp{}style relational
models into Beth\hyp{}Kripke\hyp{}style relational models.

Valliappan et al.~\cite{ValliappanRC22} give constructive
Kripke\hyp{}style relational semantics for Fitch\hyp{}style modal
lambda calculi but do not deal with positive connectives.
Dual\hyp{}context modal lambda calculi~\cite{Kavvos17,Kavvos20} do not
appear to have received a semantic treatment from the logical
perspective prior to our work and, to the best of our knowledge, does
not appear to have a constructively complete Kripke\hyp{}style
relational semantics.
Lax modal lambda calculi~\cite{Valliappan26}, corresponding to the
logics~\SL and~\LL in this article, have received a constructive
relational treatment, but also by omitting the positive connectives.
The notorious interaction between modalities and the positive
connectives appears to have received little attention in type systems
literature.

\textbf{Further Work}. In this article, our focus has been on IMLs
with a single modality.
Following Goldblatt's work on multi\hyp{}modal logics~\cite[Section
7]{Goldblatt11a}, it should be possible to extend modal cover
semantics to logics such as \CK and \CSFour, featuring both the
$\BoxPr$ and $\DiaPr$ modalities, and Fitch\hyp{}style
formulations~\cite{Borghuis94,Clouston18} that extend the logics
\CKBox and \CSFourBox with an additional modality~$\LockCtx$.

\begin{ack}

I thank Ian Shillito, Sonia Marin, Alex Kavvos, Andreas Abel, Sean
Moss, and my colleagues Justus Matthiesen and Cristina Matache at the
University of Edinburgh, for their comments and feedback on this
work. I also thank Jim de Groot for an introduction to literature on
neighborhood semantics for intuitionistic modal logics.  This work was
funded by a Royal Society Newton International Fellowship.

\end{ack}


\bibliographystyle{./entics}
\bibliography{main}


\newpage
\appendix
\section{Appendix}\label{sec:app}

\begin{figure}[h]
  \centering
  \begin{mathpar}
    \inferrule[\IPL/Hyp]{%
      A \in \Gamma
    }{%
      \Gamma \vdash A
    }

    \inferrule[\IPL/$\TopPr$\nbhyp{}Intro]{%
    }{%
      \Gamma \vdash \TopPr
    }%

    \inferrule[\IPL/$\BotPr$\nbhyp{}Elim]{%
      \Gamma \vdash \BotPr
    }{%
      \Gamma \vdash A
    }

    \inferrule[\IPL/$\AndPr$\nbhyp{}Intro]{%
      \Gamma \vdash A \\
      \Gamma \vdash B
    }{%
      \Gamma \vdash A \AndPr B
    }%

    \inferrule[\IPL/$\AndPr$\nbhyp{}Elim\nbhyp{}1]{%
      \Gamma \vdash A \AndPr B
    }{%
      \Gamma \vdash A
    }%

    \inferrule[\IPL/$\AndPr$\nbhyp{}Elim\nbhyp{}2]{%
      \Gamma \vdash A \AndPr B
    }{%
      \Gamma \vdash B
    }

    \inferrule[\IPL/$\ImpPr$\nbhyp{}Intro]{
      \ExtCtx{\Gamma}{A} \vdash B
    }{%
      \Gamma \vdash A \ImpPr B
    }%

    \inferrule[\IPL/$\ImpPr$\nbhyp{}Elim]{%
      \Gamma \vdash A \ImpPr B\\
      \Gamma \vdash A
    }{%
      \Gamma \vdash B
    }%

    \inferrule[\IPL/$\OrPr$\nbhyp{}Intro\nbhyp{}1]{
      \Gamma\vdash A
    }{%
      \Gamma \vdash A \OrPr B
    }%

    \inferrule[\IPL/$\OrPr$\nbhyp{}Intro\nbhyp{}2]{
      \Gamma\vdash B
    }{%
      \Gamma \vdash A \OrPr B
    }%

    \inferrule[\IPL/$\OrPr$\nbhyp{}Elim]{%
      \Gamma \vdash A \OrPr B\\
      \ExtCtx{\Gamma}{A} \vdash C\\
      \ExtCtx{\Gamma}{B} \vdash C
    }{%
      \Gamma \vdash C
    }%
\end{mathpar}
  \caption{Proof rules for \IPL}
\end{figure}

\begin{figure}[h]
  \centering
  \begin{mathpar}
    \inferrule[\IPL/NE/Hyp]{%
      A \in \Gamma
    }{%
      \Gamma \vdashNe A
    }

    \inferrule[\IPL/NF/$\TopPr$\nbhyp{}Intro]{%
    }{%
      \Gamma \vdashNf \TopPr
    }%

    \inferrule[\IPL/NF/$\BotPr$\nbhyp{}Elim]{%
      \Gamma \vdashNe \BotPr
    }{%
      \Gamma \vdashNf A
    }

    \inferrule[\IPL/NF/$\AndPr$\nbhyp{}Intro]{%
      \Gamma \vdashNf A \\
      \Gamma \vdashNf B
    }{%
      \Gamma \vdashNf A \AndPr B
    }%

    \inferrule[\IPL/NE/$\AndPr$\nbhyp{}Elim\nbhyp{}1]{%
      \Gamma \vdashNe A \AndPr B
    }{%
      \Gamma \vdashNe A
    }%

    \inferrule[\IPL/NE/$\AndPr$\nbhyp{}Elim\nbhyp{}2]{%
      \Gamma \vdashNe A \AndPr B
    }{%
      \Gamma \vdashNe B
    }

    \inferrule[\IPL/NF/$\ImpPr$\nbhyp{}Intro]{
      \ExtCtx{\Gamma}{A} \vdashNf B
    }{%
      \Gamma \vdashNf A \ImpPr B
    }%

    \inferrule[\IPL/NE/$\ImpPr$\nbhyp{}Elim]{%
      \Gamma \vdashNe A \ImpPr B\\
      \Gamma \vdashNf A
    }{%
      \Gamma \vdashNe B
    }%

    \inferrule[\IPL/NF/$\OrPr$\nbhyp{}Intro\nbhyp{}1]{
      \Gamma \vdashNf A
    }{%
      \Gamma \vdashNf A \OrPr B
    }%

    \inferrule[\IPL/NF/$\OrPr$\nbhyp{}Intro\nbhyp{}2]{
      \Gamma \vdashNf B
    }{%
      \Gamma \vdashNf A \OrPr B
    }%

    \inferrule[\IPL/NF/$\OrPr$\nbhyp{}Elim]{%
      \Gamma \vdashNe A \OrPr B\\
      \ExtCtx{\Gamma}{A} \vdashNf C\\
      \ExtCtx{\Gamma}{B} \vdashNf C
    }{%
      \Gamma \vdashNf C
    }%
\end{mathpar}
  \caption{Normal forms for \IPL}
\end{figure}

\begin{figure}[h]
  \centering
  \begin{mathpar}
    \inferrule[\CM/$\MonPr$\nbhyp{}Mon]{%
      \Gamma  \vdash \MonPr{A} \\
      A \vdash B
    }{%
      \Gamma \vdash \MonPr{B}
    }%

    \inferrule[\SL/$\LaxPr$\nbhyp{}Map]{%
      \Gamma  \vdash \LaxPr{A} \\
      \ExtCtx{\Gamma}{A} \vdash B
    }{%
      \Gamma \vdash \LaxPr{B}
    }%
    
    \inferrule[\LL/$\LaxPr$\nbhyp{}Intro]{%
      \Gamma  \vdash A
    }{%
      \Gamma \vdash \LaxPr{A}
    }%
    
    \inferrule[\LL/$\LaxPr$\nbhyp{}Bind]{%
      \Gamma  \vdash \LaxPr{A} \\
      \ExtCtx{\Gamma}{A} \vdash \LaxPr{B}
    }{%
      \Gamma \vdash \LaxPr{B}
    }%
\end{mathpar}

\begin{align*}
  \text{\CM} &:= \text{\IPL + \CM/$\MonPr$\nbhyp{}Mon} \\
  \text{\SL} &:= \text{\IPL + \SL/$\LaxPr$\nbhyp{}Map} \\
  \text{\LL} &:= \text{\IPL + \LL/$\LaxPr$\nbhyp{}Intro + \LL/$\LaxPr$\nbhyp{}Bind} \\
\end{align*}
  \caption{Proof rules for \CM, \SL and \LL}
\end{figure}

\begin{figure}[h]
  \centering
  \begin{mathpar}
    \inferrule[\CM/NF/$\MonPr$\nbhyp{}Mon]{%
      \Gamma  \vdashNe \MonPr{A} \\
      A \vdashNf B
    }{%
      \Gamma \vdashNf \MonPr{B}
    }

    \inferrule[\SL/NF/$\LaxPr$\nbhyp{}Map]{%
      \Gamma  \vdashNe \LaxPr{A} \\
      \ExtCtx{\Gamma}{A} \vdashNf B
    }{%
      \Gamma \vdashNf \LaxPr{B}
    }
    
    \inferrule[\LL/NF/$\LaxPr$\nbhyp{}Intro]{%
      \Gamma  \vdashNf A
    }{%
      \Gamma \vdashNf \LaxPr{A}
    }
    
    \inferrule[\LL/NF/$\LaxPr$\nbhyp{}Bind]{%
      \Gamma  \vdashNe \LaxPr{A} \\
      \ExtCtx{\Gamma}{A} \vdashNf \LaxPr{B}
    }{%
      \Gamma \vdashNf \LaxPr{B}
    }
\end{mathpar}
  \caption{Normal forms for \CM, \SL and \LL (omitting those of IPL)}
\end{figure}

\begin{figure}[h]
  \centering
  \begin{mathpar}
    \inferrule[\DCKBox/Hyp]{%
      A \in \Gamma
    }{%
      \DualCtx{\Delta}{\Gamma} \vdash A
    }

    \inferrule[\DCKBox/$\TopPr$\nbhyp{}Intro]{%
    }{%
      \DualCtx{\Delta}{\Gamma} \vdash \TopPr
    }%

    \inferrule[\DCKBox/$\BotPr$\nbhyp{}Elim]{%
      \DualCtx{\Delta}{\Gamma} \vdash \BotPr
    }{%
      \DualCtx{\Delta}{\Gamma} \vdash A
    }

    \inferrule[\DCKBox/$\AndPr$\nbhyp{}Intro]{%
      \DualCtx{\Delta}{\Gamma} \vdash A \\
      \DualCtx{\Delta}{\Gamma} \vdash B
    }{%
      \DualCtx{\Delta}{\Gamma} \vdash A \AndPr B
    }%

    \inferrule[\DCKBox/$\AndPr$\nbhyp{}Elim\nbhyp{}1]{%
      \DualCtx{\Delta}{\Gamma} \vdash A \AndPr B
    }{%
      \DualCtx{\Delta}{\Gamma} \vdash A
    }%

    \inferrule[\DCKBox/$\AndPr$\nbhyp{}Elim\nbhyp{}2]{%
      \DualCtx{\Delta}{\Gamma} \vdash A \AndPr B
    }{%
      \DualCtx{\Delta}{\Gamma} \vdash B
    }

    \inferrule[\DCKBox/$\ImpPr$\nbhyp{}Intro]{
      \DualCtx{\Delta}{\ExtCtx{\Gamma}{A}} \vdash B
    }{%
      \DualCtx{\Delta}{\Gamma} \vdash A \ImpPr B
    }%

    \inferrule[\DCKBox/$\ImpPr$\nbhyp{}Elim]{%
      \DualCtx{\Delta}{\Gamma} \vdash A \ImpPr B\\
      \DualCtx{\Delta}{\Gamma} \vdash A
    }{%
      \DualCtx{\Delta}{\Gamma} \vdash B
    }%

    \inferrule[\DCKBox/$\OrPr$\nbhyp{}Intro\nbhyp{}1]{
      \DualCtx{\Delta}{\Gamma} \vdash A
    }{%
      \DualCtx{\Delta}{\Gamma} \vdash A \OrPr B
    }%

    \inferrule[\DCKBox/$\OrPr$\nbhyp{}Intro\nbhyp{}2]{
      \DualCtx{\Delta}{\Gamma} \vdash B
    }{%
      \DualCtx{\Delta}{\Gamma} \vdash A \OrPr B
    }%

    \inferrule[\DCKBox/$\OrPr$\nbhyp{}Elim]{%
      \DualCtx{\Delta}{\Gamma} \vdash A \OrPr B\\
      \DualCtx{\Delta}{\ExtCtx{\Gamma}{A}} \vdash C\\
      \DualCtx{\Delta}{\ExtCtx{\Gamma}{B}} \vdash C
    }{%
      \Gamma \vdash C
    }

    \inferrule[\DCKBox/$\BoxPr$\nbhyp{}Intro]{%
      \DualCtx{\EmptyCtx}{\Delta} \vdash A
    }{%
      \DualCtx{\Delta}{\Gamma} \vdash \BoxPr{A}
    }
    
    \inferrule[\DCKBox/$\BoxPr$\nbhyp{}Elim]{%
      \DualCtx{\Delta}{\Gamma} \vdash \BoxPr{A} \\
       \DualCtx{\ExtCtx{\Delta}{A}}{\Gamma} \vdash B
    }{%
      \DualCtx{\Delta}{\Gamma} \vdash B
    }
\end{mathpar}
  \caption{Proof rules for \DCKBox}
\end{figure}

\begin{figure}[h]
  \centering
  \begin{mathpar}
    \inferrule[\DCKBox/NE/Hyp]{%
      A \in \Gamma
    }{%
      \DualCtx{\Delta}{\Gamma} \vdashNe A
    }

    \inferrule[\DCKBox/NF/$\TopPr$\nbhyp{}Intro]{%
    }{%
      \DualCtx{\Delta}{\Gamma} \vdashNf \TopPr
    }%

    \inferrule[\DCKBox/NF/$\BotPr$\nbhyp{}Elim]{%
      \DualCtx{\Delta}{\Gamma} \vdashNe \BotPr
    }{%
      \DualCtx{\Delta}{\Gamma} \vdashNf A
    }

    \inferrule[\DCKBox/NF/$\AndPr$\nbhyp{}Intro]{%
      \DualCtx{\Delta}{\Gamma} \vdashNf A \\
      \DualCtx{\Delta}{\Gamma} \vdashNf B
    }{%
      \DualCtx{\Delta}{\Gamma} \vdashNf A \AndPr B
    }%

    \inferrule[\DCKBox/NE/$\AndPr$\nbhyp{}Elim\nbhyp{}1]{%
      \DualCtx{\Delta}{\Gamma} \vdashNe A \AndPr B
    }{%
      \DualCtx{\Delta}{\Gamma} \vdashNe A
    }%

    \inferrule[\DCKBox/NE/$\AndPr$\nbhyp{}Elim\nbhyp{}2]{%
      \DualCtx{\Delta}{\Gamma} \vdashNe A \AndPr B
    }{%
      \DualCtx{\Delta}{\Gamma} \vdashNe B
    }

    \inferrule[\DCKBox/NF/$\ImpPr$\nbhyp{}Intro]{
      \DualCtx{\Delta}{\ExtCtx{\Gamma}{A}} \vdashNf B
    }{%
      \DualCtx{\Delta}{\Gamma} \vdashNf A \ImpPr B
    }%

    \inferrule[\DCKBox/NE/$\ImpPr$\nbhyp{}Elim]{%
      \DualCtx{\Delta}{\Gamma} \vdashNe A \ImpPr B\\
      \DualCtx{\Delta}{\Gamma} \vdashNf A
    }{%
      \DualCtx{\Delta}{\Gamma} \vdashNe B
    }%

    \inferrule[\DCKBox/NF/$\OrPr$\nbhyp{}Intro\nbhyp{}1]{
      \DualCtx{\Delta}{\Gamma} \vdashNf A
    }{%
      \DualCtx{\Delta}{\Gamma} \vdashNf A \OrPr B
    }%

    \inferrule[\DCKBox/NF/$\OrPr$\nbhyp{}Intro\nbhyp{}2]{
      \DualCtx{\Delta}{\Gamma} \vdashNf B
    }{%
      \DualCtx{\Delta}{\Gamma} \vdashNf A \OrPr B
    }%

    \inferrule[\DCKBox/NF/$\OrPr$\nbhyp{}Elim]{%
      \DualCtx{\Delta}{\Gamma} \vdashNe A \OrPr B\\
      \DualCtx{\Delta}{\ExtCtx{\Gamma}{A}} \vdashNf C\\
      \DualCtx{\Delta}{\ExtCtx{\Gamma}{B}} \vdashNf C
    }{%
      \Gamma \vdashNf C
    }

    \inferrule[\DCKBox/NF/$\BoxPr$\nbhyp{}Intro]{%
      \DualCtx{\EmptyCtx}{\Delta} \vdashNf A
    }{%
      \DualCtx{\Delta}{\Gamma} \vdashNf \BoxPr{A}
    }
    
    \inferrule[\DCKBox/NF/$\BoxPr$\nbhyp{}Elim]{%
      \DualCtx{\Delta}{\Gamma} \vdashNe \BoxPr{A} \\
       \DualCtx{\ExtCtx{\Delta}{A}}{\Gamma} \vdashNf B
    }{%
      \DualCtx{\Delta}{\Gamma} \vdashNf B
    }
  \end{mathpar}
  \caption{Normal forms for \DCKBox}
\end{figure}

\begin{figure}[h]
  \centering
  \begin{mathpar}

    \inferrule[\DCSFourBox/Hyp]{%
      A \in \Gamma
    }{%
      \DualCtx{\Delta}{\Gamma} \vdash A
    }

    \inferrule[\DCSFourBox/Global\nbhyp{}Hyp]{%
      A \in \Delta
    }{%
      \DualCtx{\Delta}{\Gamma} \vdash A
    }

    \inferrule[\DCSFourBox/$\TopPr$\nbhyp{}Intro]{%
    }{%
      \DualCtx{\Delta}{\Gamma} \vdash \TopPr
    }%

    \inferrule[\DCSFourBox/$\BotPr$\nbhyp{}Elim]{%
      \DualCtx{\Delta}{\Gamma} \vdash \BotPr
    }{%
      \DualCtx{\Delta}{\Gamma} \vdash A
    }

    \inferrule[\DCSFourBox/$\AndPr$\nbhyp{}Intro]{%
      \DualCtx{\Delta}{\Gamma} \vdash A \\
      \DualCtx{\Delta}{\Gamma} \vdash B
    }{%
      \DualCtx{\Delta}{\Gamma} \vdash A \AndPr B
    }%

    \inferrule[\DCSFourBox/$\AndPr$\nbhyp{}Elim\nbhyp{}1]{%
      \DualCtx{\Delta}{\Gamma} \vdash A \AndPr B
    }{%
      \DualCtx{\Delta}{\Gamma} \vdash A
    }%

    \inferrule[\DCSFourBox/$\AndPr$\nbhyp{}Elim\nbhyp{}2]{%
      \DualCtx{\Delta}{\Gamma} \vdash A \AndPr B
    }{%
      \DualCtx{\Delta}{\Gamma} \vdash B
    }

    \inferrule[\DCSFourBox/$\ImpPr$\nbhyp{}Intro]{
      \DualCtx{\Delta}{\ExtCtx{\Gamma}{A}} \vdash B
    }{%
      \DualCtx{\Delta}{\Gamma} \vdash A \ImpPr B
    }%

    \inferrule[\DCSFourBox/$\ImpPr$\nbhyp{}Elim]{%
      \DualCtx{\Delta}{\Gamma} \vdash A \ImpPr B\\
      \DualCtx{\Delta}{\Gamma} \vdash A
    }{%
      \DualCtx{\Delta}{\Gamma} \vdash B
    }%

    \inferrule[\DCSFourBox/$\OrPr$\nbhyp{}Intro\nbhyp{}1]{
      \DualCtx{\Delta}{\Gamma} \vdash A
    }{%
      \DualCtx{\Delta}{\Gamma} \vdash A \OrPr B
    }%

    \inferrule[\DCSFourBox/$\OrPr$\nbhyp{}Intro\nbhyp{}2]{
      \DualCtx{\Delta}{\Gamma} \vdash B
    }{%
      \DualCtx{\Delta}{\Gamma} \vdash A \OrPr B
    }%

    \inferrule[\DCSFourBox/$\OrPr$\nbhyp{}Elim]{%
      \DualCtx{\Delta}{\Gamma} \vdash A \OrPr B\\
      \DualCtx{\Delta}{\ExtCtx{\Gamma}{A}} \vdash C\\
      \DualCtx{\Delta}{\ExtCtx{\Gamma}{B}} \vdash C
    }{%
      \Gamma \vdash C
    }

    \inferrule[\DCSFourBox/$\BoxPr$\nbhyp{}Intro]{%
      \DualCtx{\Delta}{\EmptyCtx} \vdash A
    }{%
      \DualCtx{\Delta}{\Gamma} \vdash \BoxPr{A}
    }
    
    \inferrule[\DCSFourBox/$\BoxPr$\nbhyp{}Elim]{%
      \DualCtx{\Delta}{\Gamma} \vdash \BoxPr{A} \\
       \DualCtx{\ExtCtx{\Delta}{A}}{\Gamma} \vdash B
    }{%
      \DualCtx{\Delta}{\Gamma} \vdash B
    }
  \end{mathpar}
  \caption{Proof rules for \DCSFourBox}
\end{figure}

\begin{figure}[h]
  \centering
  \begin{mathpar}
    \inferrule[\DCSFourBox/NE/Hyp]{%
      A \in \Gamma
    }{%
      \DualCtx{\Delta}{\Gamma} \vdashNe A
    }

    \inferrule[\DCSFourBox/NE/Global-Hyp]{%
      A \in \Delta
    }{%
      \DualCtx{\Delta}{\Gamma} \vdashNe A
    }

    \inferrule[\DCSFourBox/NF/$\TopPr$\nbhyp{}Intro]{%
    }{%
      \DualCtx{\Delta}{\Gamma} \vdashNf \TopPr
    }%

    \inferrule[\DCSFourBox/NF/$\BotPr$\nbhyp{}Elim]{%
      \DualCtx{\Delta}{\Gamma} \vdashNe \BotPr
    }{%
      \DualCtx{\Delta}{\Gamma} \vdashNf A
    }

    \inferrule[\DCSFourBox/NF/$\AndPr$\nbhyp{}Intro]{%
      \DualCtx{\Delta}{\Gamma} \vdashNf A \\
      \DualCtx{\Delta}{\Gamma} \vdashNf B
    }{%
      \DualCtx{\Delta}{\Gamma} \vdashNf A \AndPr B
    }%

    \inferrule[\DCSFourBox/NE/$\AndPr$\nbhyp{}Elim\nbhyp{}1]{%
      \DualCtx{\Delta}{\Gamma} \vdashNe A \AndPr B
    }{%
      \DualCtx{\Delta}{\Gamma} \vdashNe A
    }%

    \inferrule[\DCSFourBox/NE/$\AndPr$\nbhyp{}Elim\nbhyp{}2]{%
      \DualCtx{\Delta}{\Gamma} \vdashNe A \AndPr B
    }{%
      \DualCtx{\Delta}{\Gamma} \vdashNe B
    }

    \inferrule[\DCSFourBox/NF/$\ImpPr$\nbhyp{}Intro]{
      \DualCtx{\Delta}{\ExtCtx{\Gamma}{A}} \vdashNf B
    }{%
      \DualCtx{\Delta}{\Gamma} \vdashNf A \ImpPr B
    }%

    \inferrule[\DCSFourBox/NE/$\ImpPr$\nbhyp{}Elim]{%
      \DualCtx{\Delta}{\Gamma} \vdashNe A \ImpPr B\\
      \DualCtx{\Delta}{\Gamma} \vdashNf A
    }{%
      \DualCtx{\Delta}{\Gamma} \vdashNe B
    }%

    \inferrule[\DCSFourBox/NF/$\OrPr$\nbhyp{}Intro\nbhyp{}1]{
      \DualCtx{\Delta}{\Gamma} \vdashNf A
    }{%
      \DualCtx{\Delta}{\Gamma} \vdashNf A \OrPr B
    }%

    \inferrule[\DCSFourBox/NF/$\OrPr$\nbhyp{}Intro\nbhyp{}2]{
      \DualCtx{\Delta}{\Gamma} \vdashNf B
    }{%
      \DualCtx{\Delta}{\Gamma} \vdashNf A \OrPr B
    }%

    \inferrule[\DCSFourBox/NF/$\OrPr$\nbhyp{}Elim]{%
      \DualCtx{\Delta}{\Gamma} \vdashNe A \OrPr B\\
      \DualCtx{\Delta}{\ExtCtx{\Gamma}{A}} \vdashNf C\\
      \DualCtx{\Delta}{\ExtCtx{\Gamma}{B}} \vdashNf C
    }{%
      \Gamma \vdashNf C
    }

    \inferrule[\DCSFourBox/NF/$\BoxPr$\nbhyp{}Intro]{%
      \DualCtx{\Delta}{\EmptyCtx} \vdashNf A
    }{%
      \DualCtx{\Delta}{\Gamma} \vdashNf \BoxPr{A}
    }
    
    \inferrule[\DCSFourBox/NF/$\BoxPr$\nbhyp{}Elim]{%
      \DualCtx{\Delta}{\Gamma} \vdashNe \BoxPr{A} \\
       \DualCtx{\ExtCtx{\Delta}{A}}{\Gamma} \vdashNf B
    }{%
      \DualCtx{\Delta}{\Gamma} \vdashNf B
    }
  \end{mathpar}
  \caption{Normal forms for \DCSFourBox}
\end{figure}


\end{document}